\documentclass[showpacs,twocolumn,aps,pra,longbibliography,superscriptaddress,notitlepage]{revtex4-2}
\usepackage{qcircuit}
\usepackage[dvips]{graphicx}
\usepackage{amsmath,amssymb,amsthm,mathrsfs,amsfonts,dsfont}
\usepackage{subfigure, epsfig}
\usepackage{braket}
\usepackage{bm}
\usepackage{enumerate}
\usepackage{algorithm}
\usepackage{algpseudocode}
\usepackage{diagbox}
\usepackage{physics}
\usepackage{color}
\usepackage{multirow}
\usepackage[marginal]{footmisc}
\usepackage{comment}
\usepackage{tikz}
\usepackage{gensymb}
\usepackage[]{qcircuit}
\usetikzlibrary{arrows}
\usetikzlibrary{shapes,fadings,snakes}
\usetikzlibrary{decorations.pathmorphing,patterns}
\usetikzlibrary{calc}
\usetikzlibrary{positioning}
\usepackage[colorlinks = true]{hyperref}

\newtheorem{theorem}{Theorem}
\newtheorem{observation}{Observation}

\newtheorem{lemma}{Lemma}

\newcommand{\mc}{\mathcal}
\newcommand{\mb}{\mathbb}
\newcommand{\id}{\mathbb{I}}

\newcommand{\comments}[1]{}

\usepackage{times}

\newcommand*{\physus}{Department of Physics, Southern University of Science and Technology, Shenzhen 518055, China}
\newcommand*{\inssus}{Shenzhen Institute for Quantum Science and Engineering, Southern University of Science and Technology, Shenzhen 518055, China}

\begin{document}

\title{Certifying Quantum Temporal Correlation via Randomized Measurements: Theory and Experiment}

\author{Hongfeng Liu}
\thanks{These authors contributed equally to this work.}
\affiliation{\physus}

\author{Zhenhuan Liu}
\thanks{These authors contributed equally to this work.}
\affiliation{Center for Quantum Information, Institute for Interdisciplinary Information Sciences, Tsinghua University, Beijing 100084, China}

\author{Shu Chen}
\affiliation{Center for Quantum Information, Institute for Interdisciplinary Information Sciences, Tsinghua University, Beijing 100084, China}

\author{Xinfang Nie}
\email{niexf@sustech.edu.cn}
\affiliation{\physus}
\affiliation{Quantum Science Center of Guangdong-Hong Kong-Macao Greater Bay Area, Shenzhen 518045, China}

\author{Xiangjing Liu} 
\email{xiangjing.liu@cnrsatcreate.sg}
\affiliation{CNRS@CREATE, 1 Create Way, 08-01 Create Tower, Singapore 138602, Singapore }
\affiliation{MajuLab, CNRS-UCA-SU-NUS-NTU International Joint Research Unit, Singapore}
\affiliation{Centre for Quantum Technologies, National University of Singapore, Singapore 117543, Singapore}

\author{Dawei Lu}
\email{ludw@sustech.edu.cn}
\affiliation{\physus}
\affiliation{Quantum Science Center of Guangdong-Hong Kong-Macao Greater Bay Area, Shenzhen 518045, China}
\affiliation{\inssus}
\affiliation{International Quantum Academy, Shenzhen 518055, China}

\begin{abstract}
We consider the certification of temporal quantum correlations using the pseudo-density matrix (PDM), an extension of the density matrix to the time domain, where negative eigenvalues are key indicators of temporal correlations. 
Conventional methods for detecting these correlations rely on PDM tomography, which often involves excessive redundant information and requires exponential resources. 
In this work, we develop an efficient protocol for temporal correlation detection by virtually preparing the PDM within a single time slice and estimating its second-order moments using randomized measurements.
Through sample complexity analysis, we demonstrate that our protocol requires only a constant number of measurement bases, making it particularly advantageous for systems utilizing ensemble average measurements, as it maintains constant runtime complexity regardless of the number of qubits. 
We experimentally validate our protocol on a nuclear magnetic resonance platform, a typical thermodynamic quantum system, where the experimental results closely align with theoretical predictions, confirming the effectiveness of our protocol.
\end{abstract}

\maketitle
\emph{\bfseries Introduction.}---Quantum correlations, both spatial and temporal, are distinguishing features of quantum mechanics. 
Over the past few decades, the utilization of quantum spatial correlations, particularly entanglement, has significantly shaped quantum information science~\cite{Horodecki2009entanglement}.
Detecting and quantifying entanglement are also essential methods for benchmarking the capabilities of quantum devices~\cite{gunhe2009entanglement,Haffner2005ion,Leibfried2005cat,monz2011fourteen,pan2012multiphoton,omran2019cat,chao2019twenty,brydges2019renyi,Cao2023super}. 
Recently, the focus on certifying quantum correlations has been generalized to include temporal correlations~\cite{Leggett1985quantum,emary2013leggett,budroni2022KS,vitagliano2023leggett,chen2024semi}.
Quantum temporal correlations, which emerge from sequential measurements on quantum systems, are crucial not only for deepening our understanding of the foundational aspects of quantum physics but also for a wide range of sequential information processing tasks. 
For example, the Leggett–Garg inequalities, derived under the assumptions of macrorealism and non-invasive measurability, can be violated by quantum mechanical predictions~\cite{Leggett1985quantum,fritz2010quantum,costantino2013bounding,chen2024semi}. 
Furthermore, temporal correlations have been employed to witness quantum dimensionality~\cite{wolf2009assessing,spee2020genuine,vieira2024witnessing} and have proven to be central in the performance of time-keeping devices~\cite{Paul2017autonomous,costantino2021clock,woods2022quantum}.

\begin{figure}[htbp]
\centering
\includegraphics[width=0.48\textwidth]{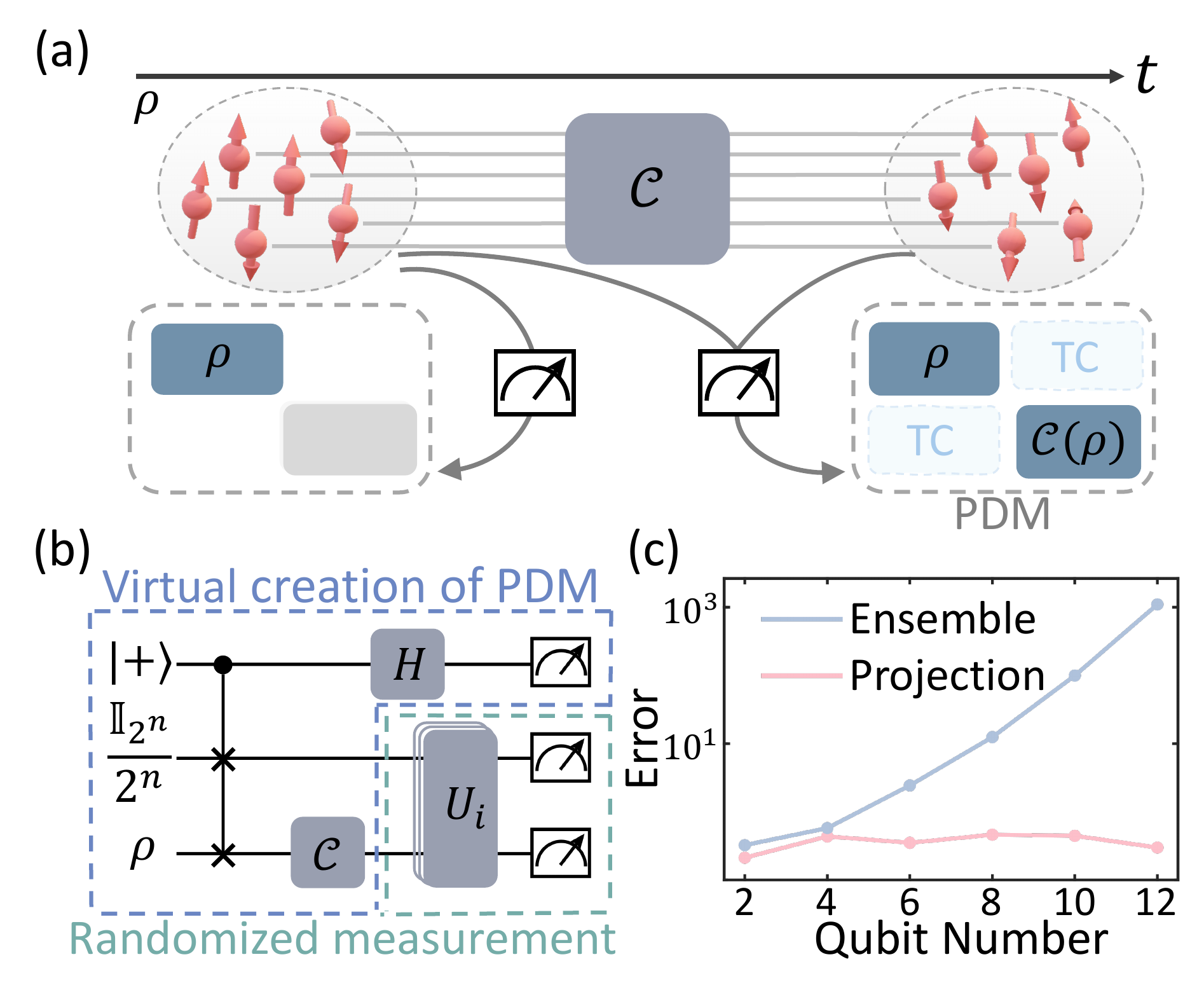}
\caption{Overview of this work. 
(a) The PDM is constructed by sequentially measuring the quantum system before and after the quantum channel, capturing the temporal correlation (TC) of the quantum process. 
(b) Our protocol consists mainly of two components. 
First, using a maximally mixed state, the input state $\rho$, a controlled SWAP gate, and a Pauli-$X$ basis measurement, the circuit within the blue box virtually prepares the PDM. 
Then, through multiple independent random unitary evolution and computational basis measurements, as shown in the green box, the purity of the PDM is estimated to infer its negativity. 
(c) The protocol is particularly suited for platforms with ensemble-average measurements. 
As illustrated, for a fixed experimental runtime, the error in ensemble-average systems does not increase, while in systems with projective measurements, the error scales exponentially with the number of qubits in the PDM. }
\label{fig:overview}
\end{figure}

Among various methodologies for certifying temporal correlations, the pseudo-density matrix (PDM) is a prominent tool due to its rich physical implications and concise mathematical form~\cite{fitzsimons2015quantum,shrotriya2022certifying,pusuluk2022witnessing,liu2024inferring,liu2023quantum,marletto2019theoretical,jia2023quantum,song2023causal,zhao2018geometry,zhang2020different,marletto2021temporal,liu2024unification}. 
As illustrated in Fig.~\ref{fig:overview}(a), the PDM is an extension of the density matrix to the time domain. 
Notably, it permits negative eigenvalues \cite{fitzsimons2015quantum}, which serve as indicators of quantum temporal correlations, as density matrices constructed from a single time point cannot exhibit such negativity. 
Additionally, compared to the violation of Leggett–Garg inequalities, negativity in the PDM functions as a subprotocol for inferring quantum causal structures~\cite{liu2023quantum,liu2024inferring} and can also be used to bound quantum channel capacity~\cite{Pisarczyk2019causal}.
The conventional method for detecting negativity involves a process known as PDM tomography~\cite{fitzsimons2015quantum,marletto2019theoretical,liu2023quantum}. 
However, as the number of qubits increases, PDM tomography becomes impractical due to the exponential consumption of both quantum and classical resources. 
Moreover, tomography can lead to redundancy when the objective is to estimate specific target information~\cite{aaronson2018shadow,huang2020predicting}, such as negativity.

In this work, we integrate the techniques of quasi-probability decomposition \cite{temme2017pec,endo2018pec} and randomized measurements \cite{Elben2019toolbox,brydges2019probing,Elben2023toolbox} to propose a protocol for detecting quantum temporal correlations without requiring full system characterization. 
Specifically, we design a quantum circuit that virtually prepares the PDM within a single time slice and employs randomized measurements to infer its negativity, as illustrated in Fig.~\ref{fig:overview}(b). 
Remarkably, our protocol requires only a constant number of measurement bases, making it particularly suitable for quantum platforms utilizing ensemble-average measurements, as depicted in Fig.\ref{fig:overview}(c). 
In such platforms, measurements in a fixed basis are performed collectively across a thermodynamically large number of copies, yielding the expectation value of an observable in a single run. 
This feature allows the runtime complexity of our protocol to be exponentially reduced compared to systems employing projective measurements.
Given the ensemble nature of nuclear spins and the high-fidelity controls achievable under ambient conditions, nuclear magnetic resonance (NMR) is an ideal candidate for testing our protocol. 
Accordingly, we conduct a proof-of-principle experiment using NMR, demonstrating that quantum temporal correlations can be efficiently detected through randomized measurements on ensemble quantum systems.

\emph{\bfseries Pseudo-density matrix.}---The PDM generalizes the concept of the density matrix to cases involving multiple time slices~\cite{fitzsimons2015quantum,liu2023quantum}. 
In this work, we focus specifically on the 2-time PDM. 
Without loss of generality, we consider scenarios where both the input and output of a quantum process are $n$-qubit states. 
Consider an $n$-qubit quantum system that is first measured by a Pauli observable $\sigma_i \in \{\mathbb{I},\sigma_x,\sigma_y,\sigma_z\}^{\otimes n}$, then transmitted through a quantum channel, followed by a second Pauli measurement $\sigma_j$. 
Let $\langle \sigma_i\sigma_j \rangle$ denote the product of the expectation values of these Pauli observables. 
Given all the expectation values $\langle \sigma_i\sigma_j \rangle$, the 2-time PDM is defined as
\begin{equation*}
R =\frac{1}{4^{n}}\sum^{4^n-1}_{i,j=0} \expval{\sigma_i\sigma_j}\sigma_i\otimes\sigma_j.
\end{equation*}
When performing the coarse-grained Pauli measurement~\cite{liu2023quantum}, the closed-formed expression of PDM coincides with the so-called canonical quantum state over time discussed in Refs.~\cite{fullwood2022quantum,arthur2023from,lie2024qsot}.
The PDM is Hermitian with unit trace but may have negative eigenvalues, which serve as indicators of quantum temporal correlations~\cite{fitzsimons2015quantum}. 

Intuitively, one could construct the PDM tomographically and then check for negative eigenvalues. 
However, the complete information obtained from tomography is excessive for this purpose and requires an impractical amount of quantum and classical resources. 
Quantitatively, if both the input and output states contain $n$ qubits, Pauli-based tomography necessitates a total of $3^{2n}$ different measurement bases. 
Furthermore, since the measurements needed to construct the PDM are more restricted than those for normal density matrices, the lower bound on the complexity of normal state tomography would still apply to PDM tomography. 
Thus, even if highly joint operations among multiple copies of the system are permitted, PDM tomography still requires $\Omega(2^{4n})$ experimental runs~\cite{haah2016tomo,chen2023adaptivity}. If only incoherent measurements are allowed, the sample complexity increases to $\Omega(2^{6n})$~\cite{chen2023adaptivity}.

\emph{\bfseries Protocol.}---To circumvent the need for PDM tomography, our core idea is to measure experimentally accessible quantities to certify temporal correlations. 
Inspired by methodologies in entanglement detection~\cite{carteret2005peres,cai2008novel,gray2018ml,elben2020mixed,zhou2020single,Neven2021resolved,yu2021optimal,liu2022detecting}, the moments of $R$, such as the purity $\Tr(R^2)$, can be used to infer the negativity of $R$.

\begin{observation} Given a PDM $R$, if $\Tr(R^2) > 1$, then $R$ is not positive semi-definite. 
\end{observation}

\noindent To prove this, note that $\Tr(R^2) = \sum_i \lambda_i^2$, where ${\lambda_i}$ are the eigenvalues of $R$, satisfying $\sum_i \lambda_i = 1$ due to $\Tr (R) = 1$. When $R \geq 0$, i.e., $\lambda_i \ge 0$ for all $i$, $\Tr(R^2)$ is upper bounded by $1$. 
Therefore, if $\Tr(R^2) > 1$, $R$ must have at least one negative eigenvalue. 
Additionally, we demonstrate in Appendix~\ref{app:PDM_purity} that $\Tr(R^2)$ can attain a maximum value of $\frac{1}{2}(2^n + 1)$ when the input is an $n$-qubit pure state and the channel is unitary. 
It can be similarly proved that the value of $\sqrt{\Tr(R^2)}$ provides a lower bound for $\Tr(|R|) = \sum_i |\lambda_i|$, which often serves as a monotone for quantum causality~\cite{fitzsimons2015quantum}.

Instead of PDM tomography, we employ the randomized measurement protocol to estimate $\Tr(R^2)$. 
It is known that the state purity, $\Tr(\rho^2)$, can be efficiently estimated by evolving $\rho$ with a random unitary and performing computational basis measurements~\cite{brydges2019probing,Elben2019toolbox}. 
However, since the PDM spans multiple time slices, estimating its moments using the randomized measurement technique is not straightforward. 
To address this challenge, we design a quantum circuit that virtually prepares the PDM within a single time slice, as shown in Fig.~\ref{fig:overview}(b).
In the circuit depicted in the blue box, we denote the measurement results of the control qubit as $0$ and $1$, their corresponding probabilities as $p_0$ and $p_1$, and the collapsing states of the other two registers as $\rho_0$ and $\rho_1$. 
We show in Appendix \ref{app:circuit} that
\begin{equation}\label{eq:virtual PDM} 
2^n(p_0\rho_0 - p_1\rho_1) = R, 
\end{equation}
where $R$ is the PDM defined by the input state $\rho$ and the channel $\mathcal{C}$. 
Note that, by setting the channel to be an identity channel, our circuit provides a new way to realize the virtual broadcasting map \cite{arthur2024virtual}.
Since the PDM virtually exists within a single time slice, we can apply the randomized measurement toolbox to estimate $\Tr(R^2)$.

As illustrated in the green box of Fig.\ref{fig:overview}(b), one begins by independently selecting $N_U$ random unitaries from a unitary ensemble $\mathcal{E}_U$ and then applying each unitary in $N_M$ independent experiments. 
For each unitary $U$, experimental data $\{s_a^i, \vec{s}^i\}_{i=1}^{N_M}$ is obtained, where $s_a$ and $\vec{s}$ denote the measurement results of the control qubit and the other two registers, respectively. 
This data is then used to construct an estimator $\hat{M}_2^U$ according to Eq.~\eqref{eq:estimator}. 
The final estimator $\hat{M}_2$ is calculated by averaging over the $N_U$ independently constructed $\hat{M}_2^U$.
In Appendix~\ref{app:RM} and Appendix~\ref{app:variance}, we analyze the estimator construction and the sample complexity and show that:

\begin{theorem}\label{theorem:estimator} 
When the unitary ensemble $\mathcal{E}_U$ is at least a unitary 2-design, and given the measurement results $\{s_a^i,\vec{s}^i\}_{i=1}^{N_M}$ obtained from the circuit in Fig.~\ref{fig:overview}(b) using the same random unitary sampled from $\mathcal{E}_U$, the expression 
\begin{equation}\label{eq:estimator} 
\hat{M}_2^U = \frac{2^{2n}}{N_M(N_M-1)} \sum_{i \neq j} (-1)^{s_a^i + s_a^j} X(\vec{s}^i, \vec{s}^j) 
\end{equation} 
is an unbiased estimator for $\Tr(R^2)$, where $X(\vec{s}^i, \vec{s}^j) = -(-2^{2n})^{\delta{\vec{s}^i, \vec{s}^j}}$.

Assume that $\hat{M}_2$ is constructed by averaging $N_U$ independent estimators $\hat{M}_2^U$. 
When the unitary ensemble $\mathcal{E}_U$ is at least a unitary 4-design, to ensure that $\abs{\hat{M}_2 - \Tr(R^2)} \le \epsilon$ with probability at least $1 - \delta$, it is required that $N_U = \mathcal{O}(\frac{1}{\epsilon^2 \delta})$ and $N_M = \mathcal{O}(2^{3n})$. 
Consequently, the total sample complexity is $N_U \times N_M = \mathcal{O}(\frac{2^{3n}}{\epsilon^2 \delta})$. 
\end{theorem}

\begin{figure*}[htbp]
\centering
\includegraphics[width=0.9\textwidth]{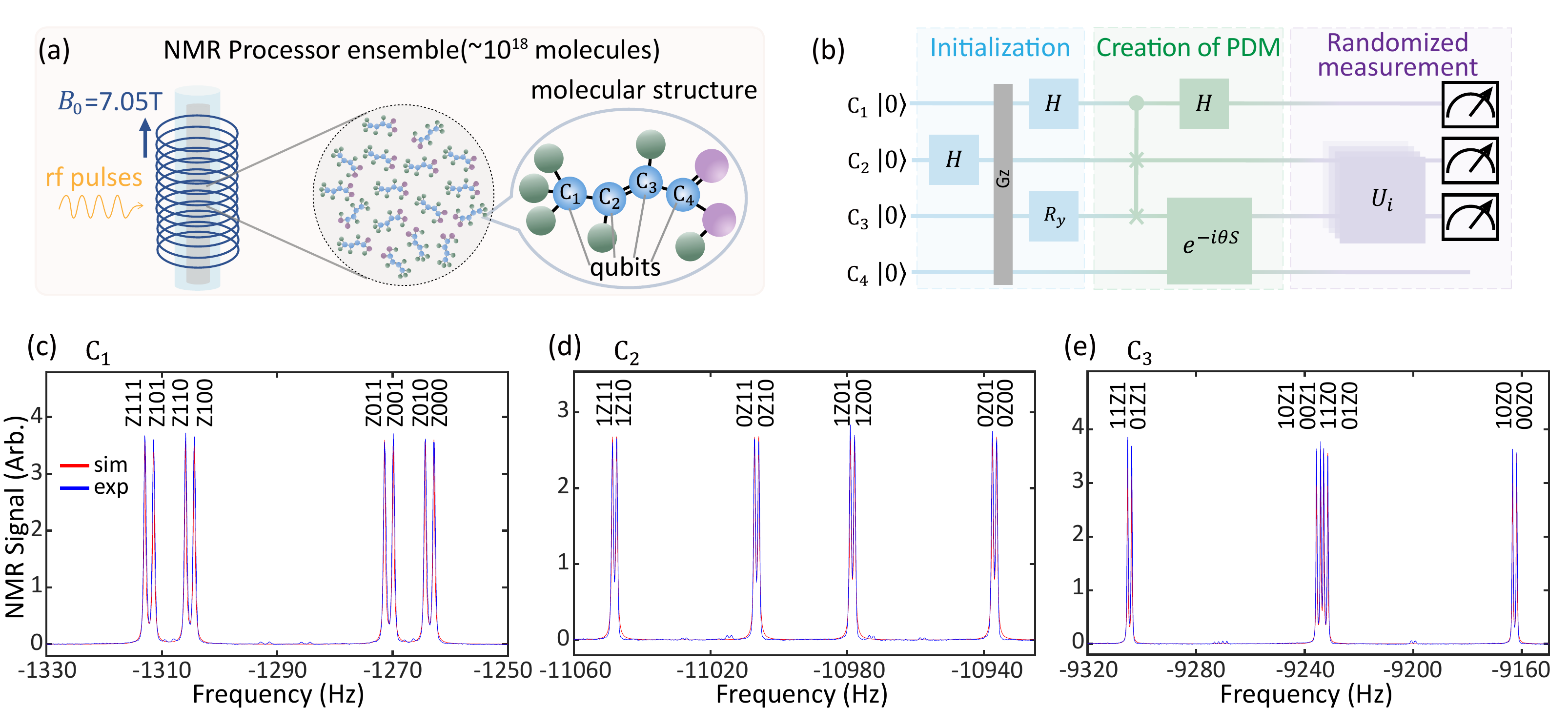}
\caption{Experimental setup. (a) The ensemble NMR system. The sample contains a large number ($\sim 10^{18}$) of identical molecules, all of which participate in the experimental process. 
(b) Quantum circuit used to implement the protocol. The circuit consists of three segments: initialization, virtual creation of the PDM, and randomized measurement. 
(c-e) Experimental spectra of $\text{C}_{1}$, $\text{C}_{2}$, and $\text{C}_{3}$, respectively. The $y$-axis has an arbitrary unit, and the absolute signal strengths in the spectra are meaningless. 
Each integral of a resonant peak at a specific frequency represents the expectation value of a particular observable. 
For example, the notation $Z111$ represents $\sigma_z \otimes \vert111\rangle\langle111\vert$.
The red lines indicate the simulation results, while the blue lines show the experimental results.
}
\label{fig:exp}
\end{figure*}

Although the sample complexity for measuring $\Tr(R^2)$ remains exponential with respect to the system size, $\mathcal{O}(2^{3n})$ is significantly smaller than the sample complexity required for PDM tomography based on joint operations. 
It is important to note that an exponential sample complexity for detecting temporal correlations is unavoidable, as the lower bounds for sample complexity in certain channel distinguishing tasks, which is achievable by measuring $\Tr(R^2)$, have been proven to be exponential~\cite{chen2022separation,chen2024tight}.
A particularly advantageous aspect of our protocol is that the number of different unitaries, i.e., the measurement bases, remains constant and independent of the system size, as indicated by $N_U = \mathcal{O}(\frac{1}{\epsilon^2\delta})$. 
This feature makes our protocol especially suitable for platforms that utilize ensemble-average measurements, such as NMR, cold atomic systems, and nitrogen vacancy centers in diamonds. 
In these platforms, the exponential projective measurements performed in a single measurement basis can be accomplished much more efficiently.

We also numerically demonstrate that this property holds even when the unitary ensemble is not a unitary 4-design. 
In Fig.\ref{fig:overview}(c), we set the channel to be a fully depolarizing channel, the input state to be $\rho = \ketbra{0}{0}$, and the unitary ensemble to be the Clifford group, which is only a unitary 3-design \cite{zhu2017clifford}. 
For the line labeled ``Ensemble", we set $N_U = 10$ and $N_M = \infty$, observing that the statistical error does not increase with the qubit number. 
For the line labeled ``Projection", we set $N_M = 100$, and it is evident that the error scales exponentially with the qubit number.

\emph{\bfseries NMR experimental scheme.}---We experimentally demonstrate our temporal correlation detection protocol on an NMR platform. 
The experiments are conducted on a Bruker 300 MHz spectrometer at room temperature using a four-qubit nuclear spin system composed of $^{13}$C-labeled trans-crotonic acid dissolved in $d_6$-acetone, with its molecular structure shown in Fig.~\ref{fig:exp}(a). 
The four carbon nuclear spins, labeled as C$_{1-4}$, form a four-qubit quantum processor~\cite{PhysRevLett.129.070502,PhysRevLett.132.210403}, characterized by the internal Hamiltonian:
$  \mathcal{H}_{\text{NMR}}=\sum^4_{i=1}\pi\nu_i\sigma^i_z+\sum^4_{i<j,=1}\frac{\pi}{2}J_{ij}\sigma^i_z\sigma^j_z,$
where $\nu_i$ represents the chemical shift of the $i$-th spin, and $J_{ij}$ denotes the scalar coupling between the $i$-th and $j$-th spins. 
Specific values for $\nu_i$ and $J_{ij}$ can be found in Appendix~\ref{app:exp}. 
In the NMR system, single-qubit rotations are implemented using transverse radio-frequency pulses, while two-qubit interactions are achieved through free evolution. The accuracy of experimental control can be further enhanced using the gradient ascent pulse engineering  algorithm~\cite{khaneja2005optimal}.

The quantum circuit used is shown in Fig.\ref{fig:exp}(b). 
The first three spins, C$_{1,2,3}$, correspond to the control qubit, the ancillary qubit, and the system qubit, respectively. 
The fourth spin, C$_4$, serves as the environment, which interacts with C$_3$ to realize the channel $\mathcal{C}$ and is not measured at the end of the circuit. 
The quantum circuit can be divided into three stages: state initialization, the virtual creation of the PDM, and the randomized measurement, as depicted in Fig.\ref{fig:exp}(b).

(i) Initialization. Starting from the $\ket{0000}$ state, we first apply a Hadamard gate to C$_2$, followed by a gradient-field pulse in the $z$-direction applied to all spins. 
This sequence transforms C$_2$ into the maximally mixed state ${\mathbb{I}_2}/{2}$ while leaving the other three qubits unchanged. 
Next, we apply another Hadamard gate to C$_1$ to prepare it in the $\ket{+}$ state and perform a rotation $R_y$ on C$_3$ to initialize the system qubit as a parameterized pure state $\ket{\psi(p)} = \sqrt{p}\ket{0} + \sqrt{1-p}\ket{1}$. 

(ii) Virtual creation of the PDM. First, we perform a controlled-SWAP gate on the first three qubits, C$_{1,2,3}$, as depicted in Fig.~\ref{fig:exp}(b). 
Next, we evolve the system spin C$_3$ and the environment spin C$_4$ using a partial SWAP operation with a parameterized time $\theta$. 
This unitary evolution induces a partial replacement channel $\mathcal{C}$ acting on C$_3$. 
Note that when $\theta = \frac{\pi}{2}$, $\mathcal{C}$ becomes a fully replacement channel that replaces any arbitrary input state with $\ketbra{0}{0}$, thereby eliminating causal influence. 
Finally, by applying a Hadamard gate to the control qubit and measuring it in the computational basis ${\ket{0}, \ket{1}}$, the residual system collapses to $\rho_0$ and $\rho_1$ with probabilities $p_0$ and $p_1$, respectively. 
At this stage, the PDM is virtually prepared, as suggested by Eq.~\eqref{eq:virtual PDM}.

(iii) Randomized measurements. The randomized measurements are performed by applying a set of $N_U = 200$ random unitaries independently sampled from the Clifford group, followed by computational basis measurements. 
As shown in Theorem~\ref{theorem:estimator}, systems with projective measurements require $N_M = \mathcal{O}(2^{3n})$ measurements for each unitary to accurately estimate $\Tr(R^2)$. 
At the same time, an NMR platform can extract all the diagonal elements of an $n$-qubit density matrix with only $n$ measurements, as discussed in Appendix~\ref{app:exp}. 
In our experiment, after the random unitary evolution, we use just three measurements to extract eight diagonal elements of the density matrix of C$_1$, C$_2$, and C$_3$.
As shown in Fig.\ref{fig:exp}(c)-(e), each integral of a resonant peak at a specific frequency represents the expectation value of an observable that is diagonal in the computational basis. 
Subsequently, the eight diagonal elements can be computed by linearly combining these expectation values. 
Denoting the diagonal elements as $\mathrm{Pr}(s_a, \vec{s})$, the estimator in Eq.~\eqref{eq:estimator} reduces to
\begin{equation*}
\hat{M}_2^U = 2^{2n} \sum_{s_a, s_a', \vec{s}, \vec{s}'} (-1)^{s_a + s_a'} X(\vec{s}, \vec{s}') \mathrm{Pr}(s_a, \vec{s}) \mathrm{Pr}(s_a', \vec{s}'), 
\end{equation*}
which corresponds to the case of $N_M = \infty$ for projective measurements. 
Therefore, the total runtime of our experiment scales only polynomially with the qubit number $n$.
Note that even with the ensemble-average measurements, performing PDM tomography remains challenging for the NMR system, as an exponential number of measurement bases is required~\cite{li2017NMRmeasurement}.

\begin{figure}[htbt]
\centering
\includegraphics[width=1\linewidth]{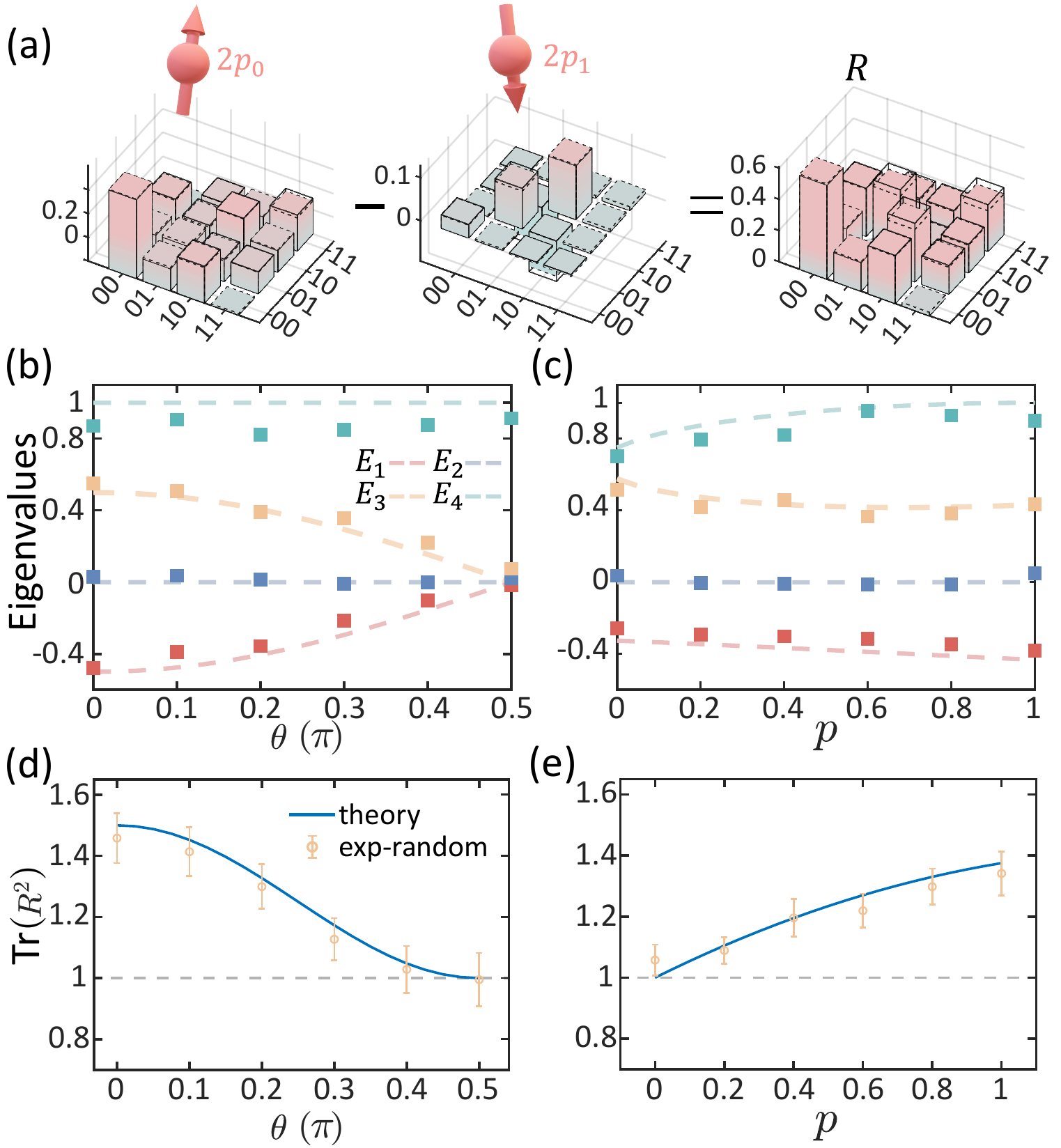}
\caption{Experimental results. (a) The calculation process and results of virtual creation of PDM.
The PDM $R$ for $\theta=\pi/6$ and $p=0.6$ is derived by subtracting the subspace where $s_a=1$ from the subspace where $s_a=0$. 
The solid bars indicate the theoretical prediction, and the colored bars with dashed lines indicate the experimental results.
(b) The eigenvalues of the corresponding PDM $R$ for different channel parameter $\theta$ and fixed state parameter $p=1$. The dashed lines are the theoretical predictions while the colored markers indicate the experimental results.
(c) The eigenvalues for different state parameters $p$ and fixed channel parameter $\theta=\pi/6$. 
(d-e) The experimental results of randomized measurement of $\Tr(R^2)$. In (d), the probability amplitude $p$ of the system state $\rho$ is fixed at 1, while in (e), $\theta$ from the channel is set at $\pi/6$. The theoretical and experimental results are shown by the blue solid line and yellow circular marker, respectively.
}
\label{fig:res}
\end{figure}

\emph{\bfseries Results.}---We first demonstrate the effectiveness of the virtual creation method. 
Before implementing the randomized measurement protocol, we experimentally performed quantum state tomography on qubits C$_{1,2,3}$ and then processed the experimental data to construct $R$, according to Eq.~\eqref{eq:virtual PDM}. 
The processing steps and results are depicted in Fig.~\ref{fig:res}(a), using the case $\theta = \pi/6$ and $p = 0.6$ as an example. 
Both the experimental predictions and theoretical results of $\rho_{0,1}$ and the PDM $R$ are presented side by side for comparison. 
The close agreement between them demonstrates the effectiveness of the virtual creation method. 
Additional tomography results can be found in Appendix~\ref{app:tomoexp}.

We also present the eigenvalues of PDMs, which include two sets of experiments:
(i) Setting $p = 1$ and varying the channel parameter $\theta$ from 0 to $\frac{\pi}{2}$, we experimentally obtain the corresponding PDM and calculate the associated eigenvalues $E_i$ of $R$, with results shown in Fig.~\ref{fig:res}(b).
(ii) Setting $\theta = \pi/6$ and varying $p$ from 0 to 1, the results are presented in Fig.~\ref{fig:res}(c). 
The experimental results are in good agreement with the theoretical predictions, where the negative eigenvalues illustrate the temporal quantum correlation.

We then applied the randomized measurement protocol to estimate $\Tr(R^2)$ and use it to certify the temporal quantum correlation. 
First, we fixed the state parameter at $p = 1$ and varied $\theta$ from 0 to $\frac{\pi}{2}$. 
The results, shown in Fig.~\ref{fig:res}(d), indicate temporal correlation for all parameters except $\theta = \pi/2$, as the estimated values are greater than one. 
The partial swap channel $e^{-i\theta S}$ reduces to the identity operation when $\theta = 0$. 
At this point, $\Tr(R^2)$ reaches its maximum value of 1.5, indicating that the final state of the system is entirely determined by its initial state. 
As $\theta$ increases, the SWAP component in the $e^{-i\theta S}$ operation becomes more pronounced, leading to a decrease in the correlation between the input and output states. 
When $\theta = \frac{\pi}{2}$, the channel acts as a full replacement channel, eliminating causality between the input and output states, resulting in $\Tr(R^2)$ being around 1.

Moreover, we fixed $\theta$ at $\pi/6$ while varying the parameter $p$ from 0 to 1 to quantify how the strength of causation depends on the initial state, with experimental results depicted in Fig.~\ref{fig:res}(e). 
As the probability amplitude $p$ increases, the measured $\Tr(R^2)$ also grows. 
All these experimental results align with theoretical expectations, demonstrating the effectiveness of our protocol.
By comparing Fig.~\ref{fig:res}(c) and Fig.~\ref{fig:res}(e), we observe that $R$ has negative eigenvalues at the point where $p = 0$, while $\Tr(R^2)$ fails to detect this negativity. 
This indicates the limitation of using only the second-order moment to detect quantum temporal correlations. 
It is worth exploring the use of higher-order moments to enhance the capabilities of our protocol~\cite{yu2021optimal}.

\emph{\bfseries Discussion.}---Since negativity in the PDM is a key signal for quantum temporal correlations, we employed quasi-probability decomposition and randomized measurements to estimate the second moment of the PDM, thereby assessing the negativity and certifying temporal quantum correlations. 
Our results naturally inspire more efficient means in temporal quantum correlation detection~\cite{emary2013leggett,vitagliano2023leggett}, quantum channel capacity~\cite{Pisarczyk2019causal} and quantum causal inference~\cite{liu2023quantum}.  
Moreover, a key finding of our study is the proficiency of NMR systems in measuring the diagonal elements of density matrices. 
It is therefore valuable to investigate other applications that could benefit from and leverage this unique capability of NMR systems.


\begin{acknowledgments}
We appreciate the valuable discussions with Andreas Elben, Richard Kueng, Oscar Dahlsten, and Mile Gu. 
HL, XN, and DL acknowledge the support from the National Key Research and Development Program of China (2019YFA0308100), National Natural Science Foundation of China (12104213,12075110,12204230), Science, Technology and Innovation Commission of Shenzhen Municipality (JCYJ20200109140803865, KQTD20190929173815000),  Guangdong Innovative and Entrepreneurial Research Team Program (2019ZT08C044), and Guangdong Provincial Key Laboratory (2019B121203002).
ZL acknowledges the support from the National Natural Science Foundation of China Grant No.~12174216 and the Innovation Program for Quantum Science and Technology Grant No.~2021ZD0300804 and No.~2021ZD0300702. 
XL is supported by the National Research Foundation, Prime Minister’s Office, Singapore under its Campus for Research Excellence and Technological Enterprise (CREATE) programme.
\end{acknowledgments}


%


\appendix

\onecolumngrid
\newpage

\section{Random Unitaries}
Here we provide some preliminaries of random unitaries. 
Intuitively, the Haar measure is the uniform distribution over a unitary group that satisfies 
\begin{equation}
\int_{\mathrm{Haar}}f(U)dU = \int_{\mathrm{Haar}}f(UV)dU =\int_{\mathrm{Haar}}f(VU)dU
\end{equation}
for arbitrary unitary $V$ and function $f(\cdot)$.
According to the Schur-Wely duality, we have
\begin{equation}
\Phi_t(X)=\int_{\mathrm{Haar}}U^{\otimes t}X U^{\dagger\otimes t}dU=\sum_{\pi,\tau\in\mc{S}_t}C_{\pi,\tau}\Tr(\hat{W}_\pi X)\hat{W}_\tau,
\end{equation}
where $\pi$ and $\tau$ are elements in the $t$-th order permutation group $\mc{S}_t$, $\hat{W}_\pi$ is the permutation operator associating with $\pi$, $C_{\pi,\tau}$ is the element of Weingarten matrix. 
In addition to the Haar measure distribution, averaging over some other unitary sets can give us the same result.
Specifically, if a unitary set $\mathcal{E}_t$ is a unitary $t$-design, then 
\begin{equation}
\int_{\mathrm{Haar}}U^{\otimes t^\prime}X U^{\dagger\otimes t^\prime}dU=\frac{1}{\norm{\mathcal{E}_t}}\sum_{U\in\mathcal{E}_t}U^{\otimes t^\prime}X U^{\dagger\otimes t^{\prime}}
\end{equation}
holds for all $t^{\prime}\le t$ and $X$.

When $t=2$, which is the case that relates to our protocol, we have
\begin{equation}\label{eq:second_twirling}
\Phi_2(X)=\frac{1}{d^2-1}\left(\Tr(X)\mb{I}-\frac{1}{d}\Tr(X)S-\frac{1}{d}\Tr(SX)\mb{I}+\Tr(SX)S\right),
\end{equation}
where $X$ is a $d^2\times d^2$ matrix, and $S$ and $\mb{I}$ are the SWAP and identity operators, which correspond to the two elements of $\mc{S}_2$. When setting $X=\sum_{s,s^\prime}X(s,s^\prime)\ketbra{s,s^\prime}{s,s^\prime}$, where $X(s,s^\prime)=-(-d)^{\delta_{s,s^\prime}}$, we have $\Phi_2(X)=S$.

\section{Circuit Analysis}\label{app:circuit}

In this section, we review a closed form of the 2-time pseudo density matrix (PDM) and give a detailed analysis of the tensor network presentation for preparing the virtual PDM.

The definition of the 2-time PDM is
\begin{equation}
R =\frac{1}{4^{n}}\sum^{4^n-1}_{i,j=0} \expval{\sigma_i\sigma_j}\sigma_i\otimes\sigma_j. 
\end{equation}
The key to constructing the PDM is to obtain the 2-time correlators $\langle \sigma_i\sigma_j\rangle$. The measurement scheme is crucial since measurements at the earlier time influence the quantum system at a later time. We call measurements that project the quantum state to the $\pm 1$ eigenspace of a $n$-qubit Pauli observable $\sigma_i$ the coarse-grained measurements. If we implement the coarse-grained measurements at each time,  then a closed form of the PDM can be written as~\cite{liu2023quantum}
\begin{align}
    R=\frac{1}{2}\left[ \Lambda_{\mc{C}} (\rho\otimes\mb{I}) + (\rho\otimes\mb{I}) \Lambda_{\mc{C}}\right],
\end{align}
where $\Lambda_{\mc{C}}$ denotes Choi–Jamiołkowski (CJ) isomorphism of the channel $\mc{C}$, which is defined as
\begin{align}
    \Lambda_{\mc{C}}=\sum_{i,j} \ket{i}\bra{j} \otimes \mc{C}( \ketbra{j}{i}).
\end{align}
We call $\Lambda_{\mc{C}}$ the CJ matrix of the channel $\mc{C}$. The action of a channel $\mc{C}$ on a quantum state $\rho$ can be equivalently given by
\begin{equation}
\mathcal{C}(\rho)=\Tr_1[\Lambda_{\mathcal{C}}(\rho\otimes\mathbb{I})].
\end{equation}
To affiliate our analysis, we re-express the PDM and $\Lambda_{\mc{C}}$ in the tensor network representation by Fig.~\ref{fig:PDM}. 
\begin{figure}[htbp]
\centering
\includegraphics[width=0.7\textwidth]{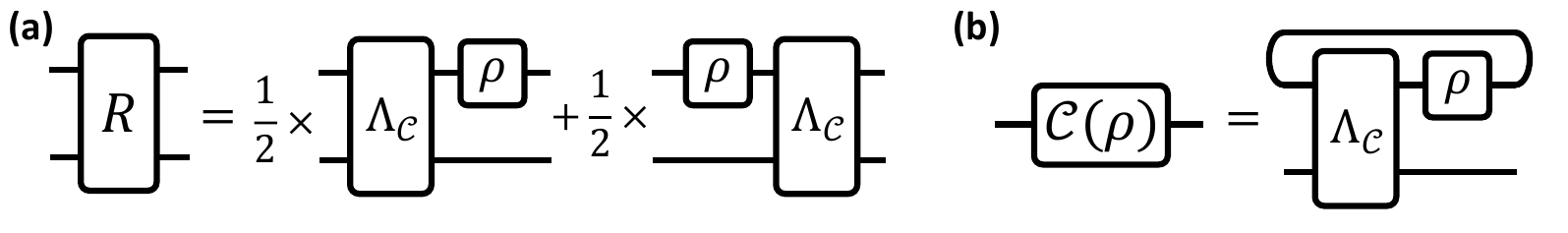}
\caption{(a) The tensor network representation of PDM using density matrix $\rho$ and CJ matrix $\Lambda_{\mathcal{C}}$. (b) The action of a channel is represented using the CJ matrix.}
\label{fig:PDM}
\end{figure}

The negativity of the PDM is the key indicator for temporal correlation. 
Based on our Observable 1, one can certify negativity in PDM via Tr$(R^2)$. Here we combine the randomized measurement technique and the circuit shown in Fig.~\ref{fig:circuit_app} to measure the moments of $R$, like Tr$(R^2)$. 
In addition to the channel $\mc{C}$ and the state $\rho$ which consist of the PDM, this circuit also contains an ancilla qubit, a maximally mixed state, the controlled SWAP operation, the random unitary evolution, and the computational basis measurements at the end. 
To show that this circuit can be used to measure the moment, we first need to analyze how the whole state evolves under this circuit.

\begin{figure}[htbp]
\centering
\includegraphics[width=0.3\textwidth]{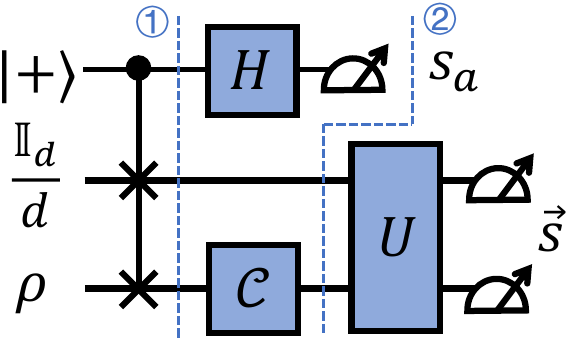}
\caption{Circuit for measuring the moment of PDM.}
\label{fig:circuit_app}
\end{figure}

At the first dashed line, the whole state evolves into 
\begin{eqnarray}
\begin{aligned}
&\frac{1}{2d}\left[  \ketbra{0}{0} \otimes (\mathbb{I} \otimes \rho)  + \ketbra{0}{1} \otimes (\mathbb{I} \otimes \rho) S + \ketbra{1}{0} \otimes S (\mathbb{I} \otimes \rho) + \ketbra{1}{1} \otimes S (\mathbb{I} \otimes \rho) S \right]\\
=&\frac{1}{2d}\left(
\ketbra{0}{0}\otimes
\begin{tabular}{c}
\includegraphics[scale=0.5]{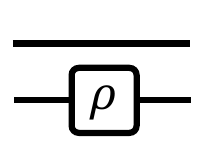}
\end{tabular}
+
\ketbra{0}{1}\otimes
\begin{tabular}{c}
\includegraphics[scale=0.5]{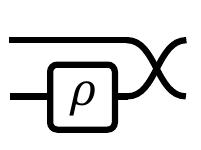}
\end{tabular}
+
\ketbra{1}{0}\otimes
\begin{tabular}{c}
\includegraphics[scale=0.5]{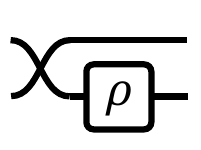}
\end{tabular}
+
\ketbra{1}{1}\otimes
\begin{tabular}{c}
\includegraphics[scale=0.5]{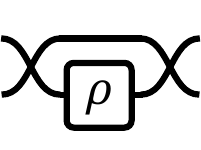}
\end{tabular}
\right)
\end{aligned}
\end{eqnarray}
After applying the channel $\mc{C}$, the whole state evolves into
\begin{eqnarray}
\begin{aligned}
&\frac{1}{2d}\left(
\ketbra{0}{0}\otimes
\begin{tabular}{c}
\includegraphics[scale=0.45]{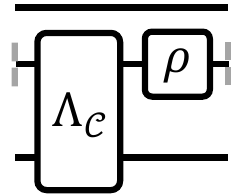}
\end{tabular}
+
\ketbra{0}{1}\otimes
\begin{tabular}{c}
\includegraphics[scale=0.45]{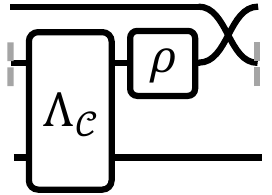}
\end{tabular}
+
\ketbra{1}{0}\otimes
\begin{tabular}{c}
\includegraphics[scale=0.45]{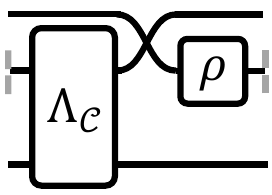}
\end{tabular}
+
\ketbra{1}{1}\otimes
\begin{tabular}{c}
\includegraphics[scale=0.45]{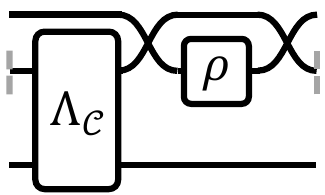}
\end{tabular}
\right)\\
=&\frac{1}{2d}\left(
\ketbra{0}{0}\otimes
\begin{tabular}{c}
\includegraphics[scale=0.45]{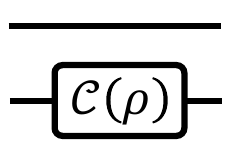}
\end{tabular}
+
\ketbra{0}{1}\otimes
\begin{tabular}{c}
\includegraphics[scale=0.45]{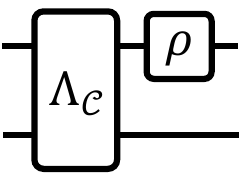}
\end{tabular}
+
\ketbra{1}{0}\otimes
\begin{tabular}{c}
\includegraphics[scale=0.45]{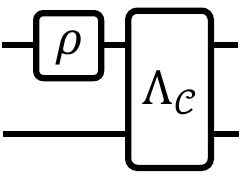}
\end{tabular}
+
\ketbra{1}{1}\otimes
\begin{tabular}{c}
\includegraphics[scale=0.45]{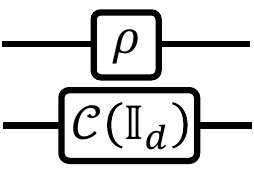}
\end{tabular}
\right),
\end{aligned}
\end{eqnarray}
where the gray dashed lines represent the trace functions.

At the second dashed line in Fig.~\ref{fig:circuit_app}, the control qubit is measured in the Pauli-$X$ basis.
Define a new matrix $Q=\frac{1}{2}\left[\mb{I}_d\otimes\mathcal{C}(\rho)+\rho\otimes\mathcal{C}(\mb{I}_d)\right]$.
Thus, the joint measurement probability distribution at the end of this circuit is
\begin{equation}\label{eq:probs}
\begin{aligned}
&\mathrm{Pr}(s_a=0,\vec{s}|U)=\frac{1}{2d}\bra{\vec{s}}U(Q+R)U^\dagger\ket{\vec{s}}\\
&\mathrm{Pr}(s_a=1,\vec{s}|U)=\frac{1}{2d}\bra{\vec{s}}U(Q-R)U^\dagger\ket{\vec{s}},
\end{aligned}
\end{equation}
where $U$ is a global random unitary acting on two systems, $s_a$ and $\vec{s}$ represent the measurement results of control qubit and the other two systems.
Here, $s_a=0$ corresponds to the control qubit collapsing to $\ket{+}$ and vice versa.
Thus, one can see that the probability distribution of the measurement results contains the information of PDM $R$. 
By taking the difference between these two probabilities, we can equivalently perform randomized measurement on the PDM $R$ and get $\bra{\vec{s}}URU^\dagger\ket{\vec{s}}$.

\section{Randomized Measurements}\label{app:RM}
We summarize our algorithm below:
\begin{algorithm}[H]
\caption{$\Tr(R^2)$ Estimation}\label{algo:causality}
\begin{algorithmic}[1]
\Require
The input state $\rho$, the quantum channel $\mathcal{C}$, a unitary ensemble $\mathcal{E}_U$ which is at least unitary two design.
\Ensure
An unbiased estimation $\hat{M}_2$ of $\Tr(R^2)$.

\For{$i= 1 \text{\textbf{ to }} N_U$} 
\State Randomly select a unitary evolution $U$ from $\mathcal{E}_U$.
\For{$j= 1 \text{\textbf{ to }} N_M$} 
\State Run the circuit shown in Fig.~\ref{fig:overview}(b) and record the measurement result $s_a^j$ and $\vec{s}^j$, which indicate measurement results of the control qubit and other qubits after the random unitary evolution.
\EndFor \,
Use the data $\{s_a^j,\vec{s}^j\}_{j=1}^{N_M}$ and Eq.~\eqref{eq:estimator} to calculate the estimator $\hat{M}_2^U$.
\EndFor
\State Calculate the average of all $\hat{M}_2^U$ to get the final estimator $\hat{M}_2$.
\end{algorithmic}
\end{algorithm}

Combining Eq.~\eqref{eq:second_twirling} and Eq.~\eqref{eq:probs}, we can now construct the estimator for $\Tr(R^2)$. 
Essentially, one applies $N_U$ independent unitaries in the circuit of Fig.~\ref{fig:circuit_app} and measures $N_M$ times under a single unitary.
After collecting all the data, the measurement results acquired with different unitaries will be processed independently.
Assume that the data $\{(s_a^i,\vec{s}^i)\}_{i=1}^{N_M}$ is collected with the same random unitary, the estimator is
\begin{equation}\label{eq:estimator_app}
\hat{M}_2^U = \frac{d^2}{N_M(N_M-1)}\sum_{i\neq j}(-1)^{s_a^i+s_a^j}X(\vec{s}^i,\vec{s}^j),
\end{equation}
where $X(s,s^\prime)=-(-d^2)^{\delta_{s,s^\prime}}$.
Here we choose $d^2$ instead of $d$ because the dimensions of PDM and the system being measured are all $d^2$.
Then, the final estimator is obtained by averaging over all estimators obtained from different unitaries.

The unbiasedness of this estimator can be verified using random unitary theory. 
Considering the independence of summation terms in Eq.~\ref{eq:estimator_app}, we have
\begin{equation}
\begin{aligned}
\mb{E}[\hat{M}_2^U] &= d^2\mb{E}_{U,s_a,s_a^\prime,\vec{s},\vec{s}^\prime}(-1)^{s_a+s_a^\prime}X(\vec{s},\vec{s}^\prime)\\
&= d^2\mb{E}_U\sum_{s_a,s_a^\prime,\vec{s},\vec{s}^\prime}(-1)^{s_a+s_a^\prime}X(\vec{s},\vec{s}^\prime) \times \mathrm{Pr}(s_a,\vec{s}|U)\times\mathrm{Pr}(s_a^\prime,\vec{s}^\prime|U)\\
&= d^2\mb{E}_U\sum_{\vec{s},\vec{s}^\prime}X(\vec{s},\vec{s}^\prime) \times \left[\mathrm{Pr}(s_a=0,\vec{s}|U)-\mathrm{Pr}(s_a=1,\vec{s}|U)\right] \times \left[\mathrm{Pr}(s_a^\prime=0,\vec{s}^\prime|U)-\mathrm{Pr}(s_a^\prime=1,\vec{s}^\prime|U)\right]\\
&= \mb{E}_U\sum_{\vec{s},\vec{s}^\prime}X(\vec{s},\vec{s}^\prime) \times \bra{\vec{s}}URU^\dagger\ket{\vec{s}} \times \bra{\vec{s}^\prime}URU^\dagger\ket{\vec{s}^\prime}\\
&= \mb{E}_U\Tr\left[R^{\otimes 2}U^{\dagger\otimes 2}XU^{\otimes 2}\right]\\
&= \Tr(R^2),
\end{aligned}
\end{equation}
where $X=\sum_{\vec{s},\vec{s}^\prime}-(-d^2)^{\delta_{\vec{s},\vec{s}^\prime}}\ketbra{\vec{s},\vec{s}^\prime}{\vec{s},\vec{s}^\prime}$, the last equality is because $\mb{E}_UU^{\otimes 2}XU^{\dagger\otimes 2}=S$ and $\Tr(S\sigma^{\otimes 2})=\Tr(\sigma^2)$.

\section{Properties of PDM}\label{app:PDM_purity}
To benefit our derivation for sample complexity, we need to have some properties of PDM.
\begin{lemma}
For arbitrary quantum channel $\mathcal{C}:\mathcal{H}_d\to\mathcal{H}_d$ and input state $\rho\in D(\mathcal{H}_d)$, we have $\Tr(R^2)\le \mathcal{O}(d)$, where the equal sign is reached by the unitary channel and pure input state.
\end{lemma}
\begin{proof}
It is easy to prove that $\Tr(R^2)=\Tr[(R^{\mathrm{T}_1})^2]$ where
\begin{equation}
R^{\mathrm{T}_1}=\frac{1}{2}\left(\Lambda_{\mathcal{C}}^{\mathrm{T}_1}(\rho\otimes\mathbb{I})+(\rho\otimes\mathbb{I})\Lambda_{\mathcal{C}}^{\mathrm{T}_1}\right)
\end{equation}
and $\mathrm{T}_1$ represents partial transposition of indices that contract with indices from $\rho$.
According to the Choi–Jamiołkowski isomorphism, $\Lambda_{\mathcal{C}}^{\mathrm{T}_1}$ is now a positive semi-definite matrix.
We then take a spectral decomposition of $\Lambda_{\mathcal{C}}^{\mathrm{T}_1}=U_1\sum_1U_1^\dagger$ and $\rho\otimes\mathbb{I}=U_2\sum_2 U_2^\dagger$, where $\sum_1$ and $\sum_2$ are positive semi-definite diagonal matrices.
Then we have
\begin{equation}
\Tr(R^2)=\frac{1}{2}\left[\Tr(U_1\Sigma_1U_1^\dagger U_2\Sigma_2U_2^\dagger U_1\Sigma_1U_1^\dagger U_2\Sigma_2U_2^\dagger)+\Tr(U_1\Sigma_1^2U_1^\dagger U_2\Sigma_2^2U_2^\dagger)\right].
\end{equation}
Defining $V=U_1^\dagger U_2$, we have
\begin{equation}
\Tr(R^2)=\frac{1}{2}\left[ \Tr(\Sigma_1V\Sigma_2V^\dagger\Sigma_1V\Sigma_2V^\dagger)+\Tr(\Sigma_1^2V\Sigma_2^2V^\dagger) \right].
\end{equation}
Defining $B=V\Sigma_2V^\dagger$, we can simplify the above expression into
\begin{equation}
\Tr(R^2)=\frac{1}{2}\left[\Tr(\Sigma_1B\Sigma_1B)+\Tr(\Sigma_1^2B^2)\right].
\end{equation}

Next we are going to prove that $\Tr(R^2)$ is a convex function of $\Sigma_1$.
We first expand matrices $\Sigma=\sum_i\lambda_i\ketbra{i}{i}$ and $B=\sum_{i,j}b_{i,j}\ketbra{i}{j}$.
Then
\begin{equation}
f(\Sigma_1)=\Tr(R^2)=\frac{1}{2}\left(\sum_{i,j}\lambda_ib_{i,j}\lambda_jb_{j,i}+\sum_{i,j}\lambda_i^2b_{i,j}b_{j,i}\right)=\frac{1}{2}\sum_{i,j}b_{i,j}b_{j,i}(\lambda_i^2+\lambda_i\lambda_j).
\end{equation}
Defining $X=\sum_ix_i\ketbra{i}{i}$ and $Y=\sum_jx_j\ketbra{j}{j}$ with $x_i,y_i\ge0$, we have
\begin{equation}
f(\theta X+(1-\theta)Y)=\frac{1}{2}\sum_{i,j}b_{i,j}b_{j,i}\left[(\theta x_i+(1-\theta)y_i)^2+(\theta x_i+(1-\theta)y_i)(\theta x_j+(1-\theta)y_j)\right]
\end{equation}
and 
\begin{equation}
\theta f(X)+(1-\theta) f(Y)=\frac{1}{2}\sum_{i,j}b_{i,j}b_{j,i}\left(\theta x_i^2+\theta x_ix_j+(1-\theta) y_i^2+(1-\theta)y_iy_j\right).
\end{equation}
The difference between them is
\begin{equation}
\begin{aligned}
&\frac{1}{2}\theta(1-\theta)\sum_{i,j}b_{i,j}b_{j,i}\left[(x_i^2+x_ix_j+y_i^2+y_iy_j)-(2x_iy_i+x_iy_j+x_jy_i)\right]\\
=&\frac{1}{2}\theta(1-\theta)\sum_{i,j}b_{i,j}b_{j,i}\left[(x_i-y_i)^2+(x_i-y_i)(x_j-y_j)\right]\\
=&\frac{1}{4}\theta(1-\theta)\sum_{i,j}b_{i,j}b_{j,i}\left[(x_i-y_i)^2+2(x_i-y_i)(x_j-y_j)+(x_j-y_j)^2\right]\\
=&\frac{1}{4}\theta(1-\theta)\sum_{i,j}b_{i,j}b_{j,i}\left[(x_i-y_i)+(x_j-y_j)\right]^2\ge 0,
\end{aligned}
\end{equation}
which shows that $\Tr(R^2)$ is a convex function of $\Sigma_1$.
Following the same logic, one can similarly prove that $\Tr(R^2)$ is also a convex function of $\Sigma_2$.
As $\Sigma_1$ and $\Sigma_2$ represent the eigenvalues of $\Lambda_{\mathcal{C}}^{\mathrm{T}_1}$ and $\rho\otimes\mathbb{I}$, they should contain at least one and $d$ positive nonzero elements, respectively.
Therefore, the convexity means that the maximal value of $\Tr(R^2)$ can be obtained when $\mathcal{C}$ is a unitary channel and $\rho$ is a pure state.

When $\mathcal{C}=\mathcal{U}$ is a unitary channel and $\rho$ is a pure state, $R^{\mathrm{T}_1}$ can be represented as
\begin{equation}
R^{\mathrm{T}_1}=\frac{1}{2}\left(\ketbra{\Psi_\mathcal{U}}{\Psi_\mathcal{U}}(\psi\otimes\mathbb{I})+(\psi\otimes\mathbb{I})\ketbra{\Psi_\mathcal{U}}{\Psi_\mathcal{U}}\right),
\end{equation}
where $\ket{\Psi_\mathcal{U}}$ is a $d^2$-dimensional unnormalized pure state satisfying $\braket{\Psi_\mathcal{U}}{\Psi_\mathcal{U}}=d$.
Then, we have 
\begin{equation}
\begin{aligned}
\Tr(R^2)=&\frac{1}{2}\left[d\bra{\Psi_\mathcal{U}}(\psi\otimes\mathbb{I})\ket{\Psi_\mathcal{U}}+\bra{\Psi_\mathcal{U}}(\psi\otimes\mathbb{I})\ket{\Psi_\mathcal{U}}^2\right]\\
=&\frac{1}{2}\left[d\Tr(U^\dagger\psi U)+\Tr(U^\dagger\psi U)^2\right]\\
=&\frac{1}{2}(d+1)=\mathcal{O}(d)
\end{aligned}
\end{equation}

\end{proof}

Note that this proof cannot be directly adopted to bound $\Tr(R^{2k})$ as $\Tr[(R^{\mathrm{T}_1})^{2k}]=\Tr(R^{2t})$ only holds when $k=1$.

\section{Variance Analysis}\label{app:variance}

Now we need to derive the sample complexity of this protocol. 
Specifically, if we want to measure $\Tr(R^2)$ to $\epsilon$ accuracy, how many experiments we need to perform? 
We consider the case that $d,N_M\gg 1$. By definition,
\begin{equation}
\mathrm{Var}(\hat{M}_2^U)=\mb{E}[(\hat{M}_2^U)^2]-\mb{E}(\hat{M}_2^U)^2.
\end{equation}
Substituting the estimator Eq.~\eqref{eq:estimator_app}, we have
\begin{equation}\label{eq:var_expression}
\mathrm{Var}(\hat{M}_2^U)=\mb{E}\left[\frac{d^4}{N_M^2(N_M-1)^2}\sum_{i\neq j}\sum_{i^\prime\neq j^\prime}(-1)^{s_a^i+s_a^j+s_a^{i^\prime}+s_a^{j^\prime}}X(\vec{s}^i,\vec{s}^j)X(\vec{s}^{i^\prime},\vec{s}^{j^\prime})\right]-\Tr(R^2)^2.
\end{equation}
Expanding the summation according to the relation between $(i,j)$ and $(i^\prime,j^\prime)$, we have
\begin{equation}
\begin{aligned}
&\mb{E}\left[\sum_{i\neq j}\sum_{i^\prime\neq j^\prime}(-1)^{s_a^i+s_a^j+s_a^{i^\prime}+s_a^{j^\prime}}X(\vec{s}^i,\vec{s}^j)X(\vec{s}^{i^\prime},\vec{s}^{j^\prime})\right]\\
=&\mb{E}\left[\sum_{i\neq j}X(\vec{s}^i,\vec{s}^j)^2+\sum_{i\neq j\neq k}(-1)^{s_a^i+s_a^k}X(\vec{s}^i,\vec{s}^j)X(\vec{s}^j,\vec{s}^k)+\sum_{i\neq j\neq k\neq l}(-1)^{s_a^i+s_a^j+s_a^k+s_a^l}X(\vec{s}^i,\vec{s}^j)X(\vec{s}^k,\vec{s}^l)\right]\\
=& 2N_M(N_M-1)\mb{E}\left[X(\vec{s},\vec{s}^\prime)^2\right]+4N_M(N_M-1)(N_M-2)\mb{E}\left[(-1)^{s_a+s_a^{\prime\prime}}X(\vec{s},\vec{s}^\prime)X(\vec{s}^\prime,\vec{s}^{\prime\prime})\right]\\
&+N_M(N_M-1)(N_M-2)(N_M-3)\mb{E}\left[(-1)^{s_a+s_a^\prime+s_a^{\prime\prime}+s_a^{\prime\prime\prime}}X(\vec{s},\vec{s}^\prime)X(\vec{s}^{\prime\prime},\vec{s}^{\prime\prime\prime})\right].
\end{aligned}
\end{equation}

We now need to specify the three terms one by one. Firstly,
\begin{equation}
\begin{aligned}
\mb{E}\left[X(\vec{s},\vec{s}^\prime)^2\right]=&\mb{E}_U\left\{\sum_{\vec{s},\vec{s}^\prime}\left[\mathrm{Pr}(s_a=0,\vec{s}|U)+\mathrm{Pr}(s_a=1,\vec{s}|U)\right]\left[\mathrm{Pr}(s_a^\prime=0,\vec{s}^\prime|U)+\mathrm{Pr}(s_a^\prime=1,\vec{s}^\prime|U)\right]X(\vec{s},\vec{s}^\prime)^2\right\}\\
=&\frac{1}{d^2}\mb{E}_U\left\{\sum_{\vec{s},\vec{s}^\prime}\bra{\vec{s}}UQU^\dagger\ket{\vec{s}}\bra{\vec{s}^\prime}UQU^\dagger\ket{\vec{s}^\prime}X(\vec{s},\vec{s}^\prime)^2\right\}\\
=&\frac{1}{d^2}\mb{E}_U\Tr\left[U^{\dagger\otimes 2}X^2U^{\otimes 2}Q^{\otimes 2}\right]\\
=&\frac{1}{d^2}\Tr\left\{[d^2\mb{I}+(d^2-1)S]Q^{\otimes 2}\right\}\\
=&\Tr(Q)^2+\frac{d^2-1}{d^2}\Tr(Q^2),
\end{aligned}
\end{equation}
where the fourth equal sign is because $\mb{E}_U(U^{\dagger\otimes 2}X^2U^{\otimes 2})=d^2\mb{I}+(d^2-1)S$, which can be verified using the random unitary theory. 
By definition, $Q=\frac{1}{2}\left[\mb{I}_d\otimes\mathcal{C}(\rho)+\rho\otimes\mathcal{C}(\mb{I}_d)\right]$.
Thus, we have $\Tr(Q)=d$ and $\Tr(Q^2)=\frac{d^2}{4}\Tr\left[\mathcal{C}(\frac{\mathbb{I}_d}{d})^2\right]\Tr(\rho^2)+\frac{d}{4}\Tr\left[\mc{C}(\rho)^2\right]+\frac{d}{2}\Tr\left[\mathcal{C}(\rho)\mathcal{C}(\frac{\mathbb{I}_d}{d})\right]$.
Substituting this into the calculation of the expectation, we get the first term
\begin{equation}\label{eq:var_term1}
\mb{E}\left[X(\vec{s},\vec{s}^\prime)^2\right]=d^2+\frac{d^2-1}{d^2}\left\{\frac{d^2}{4}\Tr\left[\mathcal{C}(\frac{\mathbb{I}_d}{d})^2\right]\Tr(\rho^2)+\frac{d}{4}\Tr\left[\mc{C}(\rho)^2\right]+\frac{d}{2}\Tr\left[\mathcal{C}(\rho)\mathcal{C}(\frac{\mathbb{I}_d}{d})\right]\right\}\le\frac{5}{4}d^2+\mc{O}(d).
\end{equation}

Similarly, for the second term, we have
\begin{equation}
\begin{aligned}
&\mb{E}\left[(-1)^{s_a+s_a^{\prime\prime}}X(\vec{s},\vec{s}^\prime)X(\vec{s}^\prime,\vec{s}^{\prime\prime})\right]\\
=&\frac{1}{d^3}\mb{E}_U\left\{\sum_{\vec{s},\vec{s}^\prime,\vec{s}^{\prime\prime}}\bra{\vec{s}}URU^\dagger\ket{\vec{s}}\bra{\vec{s}^\prime}UQU^\dagger\ket{\vec{s}^\prime}\bra{\vec{s}^{\prime\prime}}UQU^\dagger\ket{\vec{s}^{\prime\prime}}X(\vec{s},\vec{s}^\prime)X(\vec{s}^\prime,\vec{s}^{\prime\prime})\right\}\\
=&\frac{1}{d^3}\mb{E}_U\Tr\left[U^{\dagger\otimes 3}X_3U^{\otimes 3}(R\otimes Q\otimes R)\right]\\
=&-\frac{1}{d^3(d^2+2)}\left[\Tr(Q)\Tr(R)^2+2\Tr(R)\Tr(QR)\right]+\frac{d^2+1}{d^3(d^2+2)}\left[\Tr(Q)\Tr(R^2)+2\Tr(QR^2)\right],
\end{aligned}
\end{equation}
where $X_3=\sum_{\vec{s},\vec{s}^\prime,\vec{s}^{\prime\prime}}X(\vec{s},\vec{s}^\prime)X(\vec{s}^\prime,\vec{s}^{\prime\prime})\ketbra{\vec{s},\vec{s}^\prime,\vec{s}^{\prime\prime}}{\vec{s},\vec{s}^\prime,\vec{s}^{\prime\prime}}$ and the proof of the last equal sign can be found in Ref.~\cite{liu2022detecting}.
Using facts including $\Tr(Q)=d$, $\Tr(R)=1$, $\Tr(Q^2)\le\mc{O}(d^2)$, $\Tr(R^4)\le\Tr(R^2)^2\le d^2$, $\abs{\Tr(QR)}\le\sqrt{\Tr(Q^2)\Tr(R^2)}\le\mc{O}(d^{3/2})$, and $\Tr(QR^2)\le\Tr(Q)\Tr(R^2)\le\mc{O}(d^2)$, we have 
\begin{equation}\label{eq:var_term2}
\mb{E}\left[(-1)^{s_a+s_a^{\prime\prime}}X(\vec{s},\vec{s}^\prime)X(\vec{s}^\prime,\vec{s}^{\prime\prime})\right]\le\mc{O}\left(\frac{1}{d}\right).
\end{equation}

Combining the conclusions derived in Ref.~\cite{liu2022detecting} and $\Tr(R), \Tr(R^2), \Tr(R^4)>0$, we can prove that the next term is 
\begin{equation}
\begin{aligned}
&\mb{E}\left[(-1)^{s_a+s_a^\prime+s_a^{\prime\prime}+s_a^{\prime\prime\prime}}X(\vec{s},\vec{s}^\prime)X(\vec{s}^{\prime\prime},\vec{s}^{\prime\prime\prime})\right]\\
=&\frac{1}{d^4}\mb{E}_U\Tr[U^{\dagger\otimes 4}X^{\otimes 2}U^{\otimes 4}R^{\otimes 4}]\\
\le&\frac{2}{d^6(d^2+2)(d^2+3)}\left[1+2\Tr(R^2)\right]+\frac{8(d^2+1)}{d^6(d^2+2)(d^2+3)}\abs{\Tr(R^3)}\\
&+\frac{d^2+1}{d^6(d^2+3)}\left[2\Tr(R^2)^2+2\Tr(R^4)\right]+\frac{d^2(d^2+2)(d^2+3)+2}{d^6(d^2+2)(d^2+3)}\Tr(R^2)^2.
\end{aligned}
\end{equation}
As $\abs{\Tr(R^3)}\le\sqrt{\Tr(R^2)\Tr(R^4)}\le\mathcal{O}(d^{3/2})$, we have
\begin{equation}\label{eq:var_term3}
\mb{E}\left[(-1)^{s_a+s_a^\prime+s_a^{\prime\prime}+s_a^{\prime\prime\prime}}X(\vec{s},\vec{s}^\prime)X(\vec{s}^{\prime\prime},\vec{s}^{\prime\prime\prime})\right]\le\mathcal{O}\left(\frac{1}{d^4}\right)+\frac{\Tr(R^2)^2}{d^4}
\end{equation}

Substituting Eq.~\eqref{eq:var_term1}, Eq.~\eqref{eq:var_term2}, and Eq.~\eqref{eq:var_term3} into Eq.~\eqref{eq:var_expression}, we have
\begin{equation}
\mathrm{Var}(\hat{M}_2^U)\le\mathcal{O}\left(\frac{d^6}{N_M^2}+\frac{d^3}{N_M}+1\right).
\end{equation}
As the final estimator is obtained by averaging over data collected in $N_U$ different unitaries, the total variance is
\begin{equation}
\mathcal{O}\left[\frac{1}{N_U}\left(\frac{d^6}{N_M^2}+\frac{d^3}{N_M}+1\right)\right].
\end{equation}
Thus, according to Chebyshev's inequality, to make sure that $\abs{\hat{M}_2-\Tr(R^2)}\le\epsilon$ with probability at least $1-\delta$, the experiment complexity should satisfy $N_M=\mathcal{O}(d^3)$ and $N_U=\mathcal{O}(\frac{1}{\epsilon^2\delta})$.
The total sample complexity is $N_M\times N_U=\mathcal{O}(\frac{d^3}{\epsilon^2\delta})$

As discussed in the main context, for NMR platform, it is equivalent to the case of $N_M=\infty$ as every computational basis measurement is performed on an ensemble with particle number of the thermodynamic limit.
Then, the total sample complexity is equivalent to the number of different unitaries $N_U=\mathcal{O}(\frac{1}{\epsilon^2\delta})$, which is independent of the system size.

\section{Theoretical model in our experiment}

We introduce the physical process and the corresponding PDM of our experiment. In the main body, the physical process involves a quantum state  $\ket{\psi}=\sqrt{p}\ket{0}+\sqrt{1-p}\ket{1}$
undergoing a partial swap interaction with the environment qubit $\gamma_E= \ket{0}\bra{0}$.  To simplify our analysis, we take $p=1$ in the following.

Consider a system qubit $\rho_S= \ketbra{0}{0}$ interacts with an environment qubit $\gamma_E= \ketbra{0}{0}$ via the partial swap interaction 
$$
V=e^{-i \theta S}= \cos(\theta) \,  \mathbb{I} + i  \sin(\theta) \, S := c \mathbb{I} + i s S,
$$
where $S$ denotes the 2-qubit swap operator and $\theta\in [0,\frac{\pi}{2}]$. 
The effective dynamics of the system can be modeled as a partial replacement channel $\mc{N}$ with the set of Kraus operators given by 
\begin{align}\label{eq:krausop}
\{K_1 = c \mathbb{I} + i  s \ket{0} \bra{0}, K_2= is \ket{0}\bra{1}  \}.
\end{align}
One sees that when $c=0$, the initial state of the system cannot influence its final state, i.e., no temporal correlation from the input to the output. However, the influence from its input to output increases when $c$ increases.

Given the Kraus operators of the channel $\mc{N}$, one can calculate the CJ matrix $\Lambda_{\mc{N}}$. Therefore, the corresponding PDM of the system across two times is given by 
\begin{align}
R= \frac{1}{2}\left[ (\rho_s \otimes \id) \, \Lambda_{\mc{N}} + \Lambda_{\mc{N}} \, (\rho \otimes \id) \right] =
    \begin{pmatrix}
        1 & 0  & 0  &  0  \\
        0 & 0  & \frac{c(c + i s)}{2}  &  0  \\
        0 &  \frac{c(c - i s)}{2} &0  &  0  \\
        0 &  0 &0  &  0  \\
    \end{pmatrix}.
\end{align}
We proceed to obtain
\begin{align}
\Tr(R^2)-1 = \frac{1}{2}c^2.
\end{align}
The quantity $\Tr(R^2)-1$ is equal to 0 when $c=0$, and monotonically increases as $c$ increase.  Hence, $\Tr(R^2)$ correctly characterizes the temporal correlation (causal influence) from the system's input to its output.

\section{Experimental Details}\label{app:exp}
Our experiments were conducted using a nuclear magnetic resonance (NMR) quantum processor, which utilizes the nuclear spins within a molecule to encode qubits. Initially, we provide a detailed characterization of the NMR system, covering aspects such as sample preparation, control mechanisms, and measurement techniques. Subsequently, we will elaborate on the method for creating a pseudo-pure state (PPS). Lastly, we describe the process of quantum state tomography in NMR systems.

\textit{Characterization}--In this experiment, the $^{13}$C-labeled trans-crotonic acid dissolved in $d_6$ acetone is used as a 4-qubit quantum processor. The molecular structure and the relevant parameters is shown in Fig.~\ref{molecular structrue}, where $^{13}$C$_1$ to $^{13}$C$_4$ correspond to qubits Q1 to Q4. The methyl group M and all hydrogen atoms were decoupled throughout all experiments. The total Hamiltonian $\mathcal{H}_\text{tot}$ of this system includes the internal Hamiltonian $\mathcal{H}_\text{int}$ and the control Hamiltonian $\mathcal{H}_\text{con}$ is
\begin{align}
{\mathcal{H}_{{\rm{tot}}}} =&{\mathcal{H}_{{\rm{int}}}} +{\mathcal{H}_{{\rm{con}}}} \\\nonumber
 =&\sum\limits_{i = 1}^4 {\pi {\nu_i}\sigma _z^i} + \sum\limits_{1 \le i < j \le 4}^4 {\frac{\pi }{2}{J_{ij}}\sigma _z^i} \sigma _z^j\\\nonumber
 &-B_1\sum\limits_{i = 1}^4 \gamma_i[\cos(\omega_{rf}t+ \phi)\sigma_x^i+\sin(\omega_{rf}t+ \phi)\sigma_y^i],
\end{align}
where the $\nu_i$ is the chemical shift of the $i$th spin and $J_{ij}$ is the scalar coupling strength between spins $i$-th and $j$-th nuclei. Here, $B_1$, $\omega_{rf}$ and $\phi$ denote the amplitude, frequency and phase of the control pulse, respectively.

\begin{figure}
\centering
\includegraphics[width=0.7\linewidth]{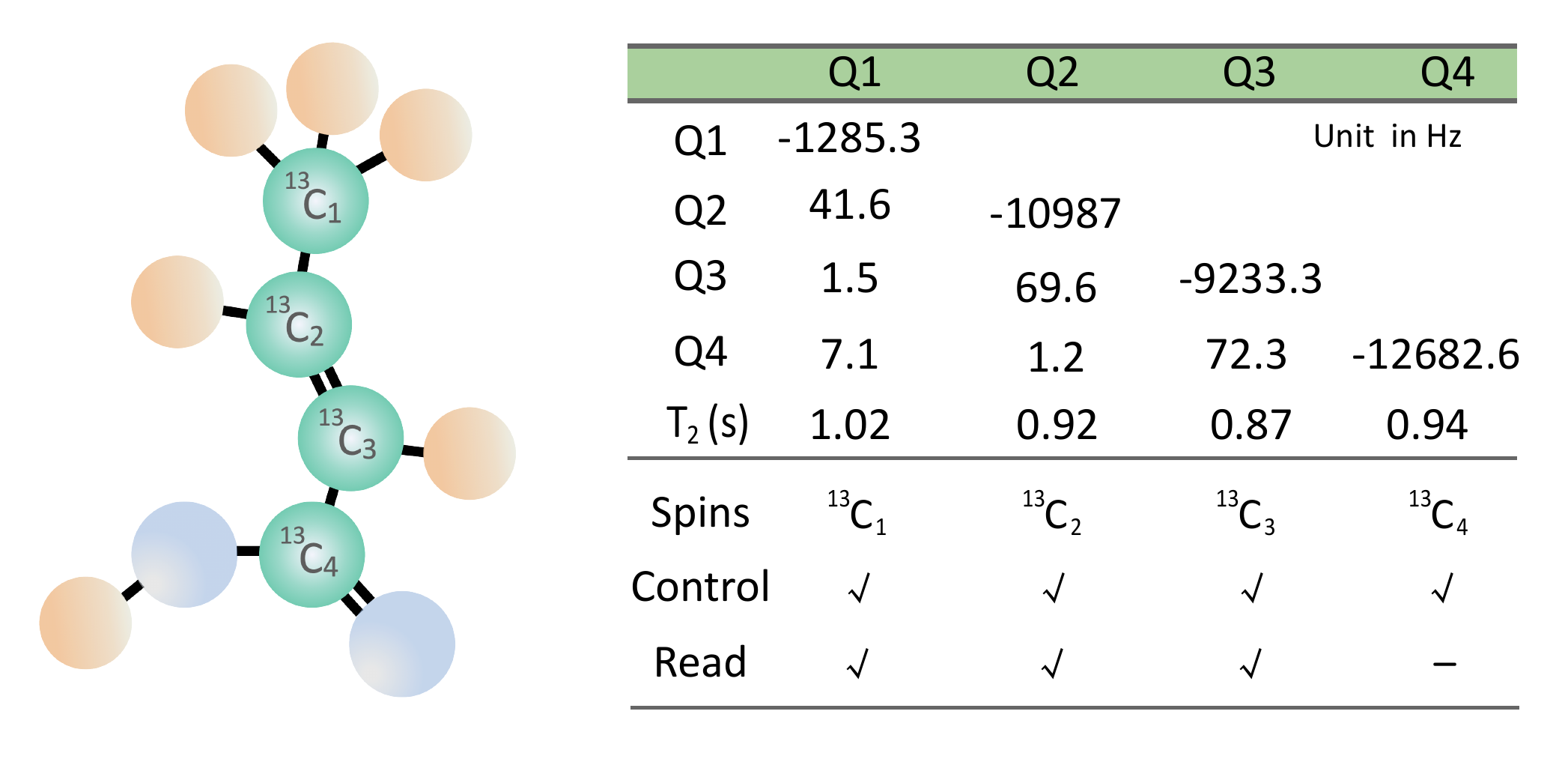}
\caption{The molecular structure and NMR parameters of $^{13}$C-labeled trans-crotonic acid. In the table, the chemical shifts and the scalar coupling constants (in Hz) are listed by the diagonal and off-diagonal numbers, respectively. The relaxation time $T_2$ (in seconds) are shown at the table.}
\label{molecular structrue}
\end{figure}

\textit{Pseudo-pure state preparation.}--At room temperature, the thermal equilibrium state of the four-qubit NMR system is a highly mixed state that described by
\begin{equation}
\rho_{\text{eq}} = \frac{\mathbb{I}}{16} + \epsilon\sum^4_{i=1}\sigma_z^i,
\label{eqstate}
\end{equation}
where $\mathbb{I}$ is the $16 \times 16$ identity matrix, and $\epsilon$, representing polarization, is approximately $10^{-5}$. This state is unsuitable for use as the initial state in quantum computing.
Various initialization methods are available, including the spatial averaging method, line-selective transition method, time averaging method, and cat-state method. In our experiments, we employed the spatial averaging method to initialize the NMR system, applying the pulse sequence illustrated in Fig.\ref{pps&measure}(a). In the circuit diagram, colored rectangles indicate single-qubit rotations performed using rf pulses. The two-qubit gates are achieved through the scalar coupling among different spins combined with shaped pulses. The pulse sequence transforms the equilibrium state described in Eq.~\ref{eqstate} into a pseudo-pure state (PPS), as shown by the equation
\begin{equation}
\rho_{\text{PPS}} = \frac{1-\epsilon'}{16}\mathbb{I} + \epsilon'\op{0000}{0000}.
\label{eq:PPS}
\end{equation}
The dominant component, the identity matrix $\mathbb{I}$, remains constant under any unitary transformation and is undetectable in NMR experiments. This characteristic enables the quantum system to be effectively treated as the pure state $\op{0000}{0000}$, despite its actual mixed nature. In our experimental setup, we combined each segment of the quantum circuit, separated by three gradient pulses, into one unitary operation. We then utilized the optimal-control algorithm to search for the corresponding rf pulse. The shaped pulses used in the experiments had lengths of 3 ms, 20 ms, 15 ms and 15 ms, respectively. All the pulses exhibited fidelities over 99.5$\%$.

\textit{Measurement.}--In the NMR quantum processor, the experimental sample consists not of a single molecule, but rather of a system comprising a large number of identical molecules. Consequently, the measurements performed by the NMR system are ensemble averages. After the operation, the nuclear spins precess around the $B_0$ direction and gradually return to thermal equilibrium. The precessing nuclear spins induce an electrical signal in the $x,y$-plane. Thus, the NMR system can only measure the transverse magnetization vectors, specifically the expectation values of $\sigma_x$ and $\sigma_y$. In a four-qubit NMR quantum processor, the signal of each spin is usually split into 8 peaks due to the couplings between different nuclei. 
According to the spin dynamics in NMR, the signal of each peak includes both real and imaginary components. These components encode the expectation values of the Pauli matrices $\sigma_x$ and $\sigma_y$ for the observable spin, respectively.
Consequently, NMR can measure the expectation values of the single-quantum coherence operators consisted of $\sigma_x$ or $\sigma_y$ in the target qubit and $\sigma_z$ or $I$ in the rest qubits. In our protocol, the focus is on measuring longitudinal magnetization observables such as $\sigma_zIII$. To facilitate this, readout pulses are applied to convert these observables into their transverse counterparts. Specifically, the readout pulse $R_y^1(\pi/2)$ is used to measure $\sigma_zIII$ by transferring it to $\sigma_xIII$.

\begin{figure}
\centering
\includegraphics[width=1\linewidth]{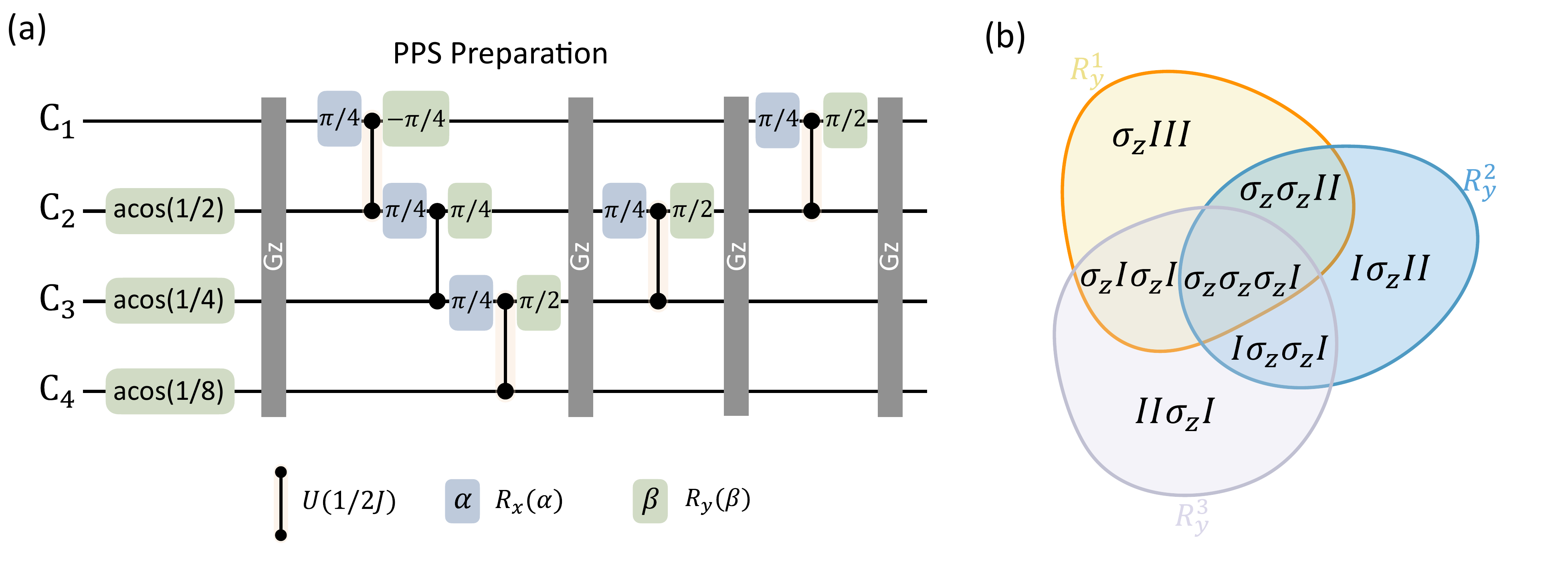}
\caption{(a) NMR pulse sequence to prepare the 4-qubit system to the PPS. The blue and green rectangles indicate the $R_x$ and $R_y$ rotation gate, respectively. The grey rectangles mean the Gz pulse. The gradient-field pulse, denoted by Gz, are used to eliminate all coherence from the instantaneous state. (b) 
Venn diagram visualizing the target observables set for our randomized meassurement task. The readout pulse for measuring the operators to perform the reduced 3-qubit diagonal density matrix tomography. Each circle represents the obervables that can be measured after applying a readout pulse of the corresponding color. }
\label{pps&measure}
\end{figure}

\textit{Quantum state tomography}--To measuring the $\text{Tr}(R^2)$ using randomization measurements, only the diagonal elements of the density matrix of final state are required. Here, we illustrate the process of performing tomography on the diagonal elements of a density matrix. The diagonal elements can be decomposed via the Pauli basis $\prod^{n}_{i}\otimes\sigma^i_{0,z}$, where the Pauli matrices $\sigma_0=I$ and $\sigma_z$ are used. Hence, using the above readout method, the diagonal elements of an unknown quantum state $\rho$ can be determined. 

In our experiment, we focus exclusively on the final state of the first three qubits and perform direct measurements on them to obtain the expectation values of the 8 Pauli operators $\sigma_{0,z}\otimes\sigma_{0,z}\otimes\sigma_{0,z}\otimes I$. As shown in Fig.~\ref{pps&measure}(b), three readout operations are employed to realize the tomography of the reduced 3-qubit diagonal density matrix. Specifically, the figure only shows a subset of the measurable observables that we are concerned with. In reality, the number of observables that can be measured after each readout pulse far exceeds this subset.

\begin{table}
\begin{tabular}{ccc}
\toprule
Experiments & Readout Pulses &  Measured Observables\\
\hline
      & $I$    & $\sigma_xIII$, $\sigma_x\sigma_zII$, $\sigma_xI\sigma_zI$, $\sigma_x\sigma_z\sigma_zI$,\\ 
      &        & $I\sigma_xII$, $I\sigma_x\sigma_zI$, $\sigma_z\sigma_xII$, $\sigma_z\sigma_x\sigma_zI$,\\
      &        & $II\sigma_xI$, $I\sigma_z\sigma_xI$, $\sigma_zI\sigma_xI$, $\sigma_z\sigma_z\sigma_xI$\\ 
No. 1 & $R_y^1(\pi/2)$ & $\sigma_zIII$, $\sigma_z\sigma_zII$, $\sigma_zI\sigma_zI$, $\sigma_z\sigma_z\sigma_zI$\\
No. 2 & $R_y^2(\pi/2)$ & $I\sigma_zII$, $I\sigma_z\sigma_zI$, $\sigma_z\sigma_zII$, $\sigma_z\sigma_z\sigma_zI$\\
No. 3 & $R_y^3(\pi/2)$ & $II\sigma_zI$, $I\sigma_z\sigma_zI$, $\sigma_zI\sigma_zI$, $\sigma_z\sigma_z\sigma_zI$\\

\toprule
\end{tabular}
\caption{Readout pulses and the measured operators to perform the reduced 3-qubit diagonal density matrix quantum state tomography.}
\label{reduced}
\end{table}

\section{Experimental Tomography Result for $R$}\label{app:tomoexp}
The results of the PDM $R$ are shown in the Fig~.\ref{exptomores1} and Fig~.\ref{exptomores2}. The solid lines represent the theoretical predication while the color bars with dashed lines indicate the experimentally results.

\begin{figure}
\centering
\includegraphics[width=0.9\linewidth]{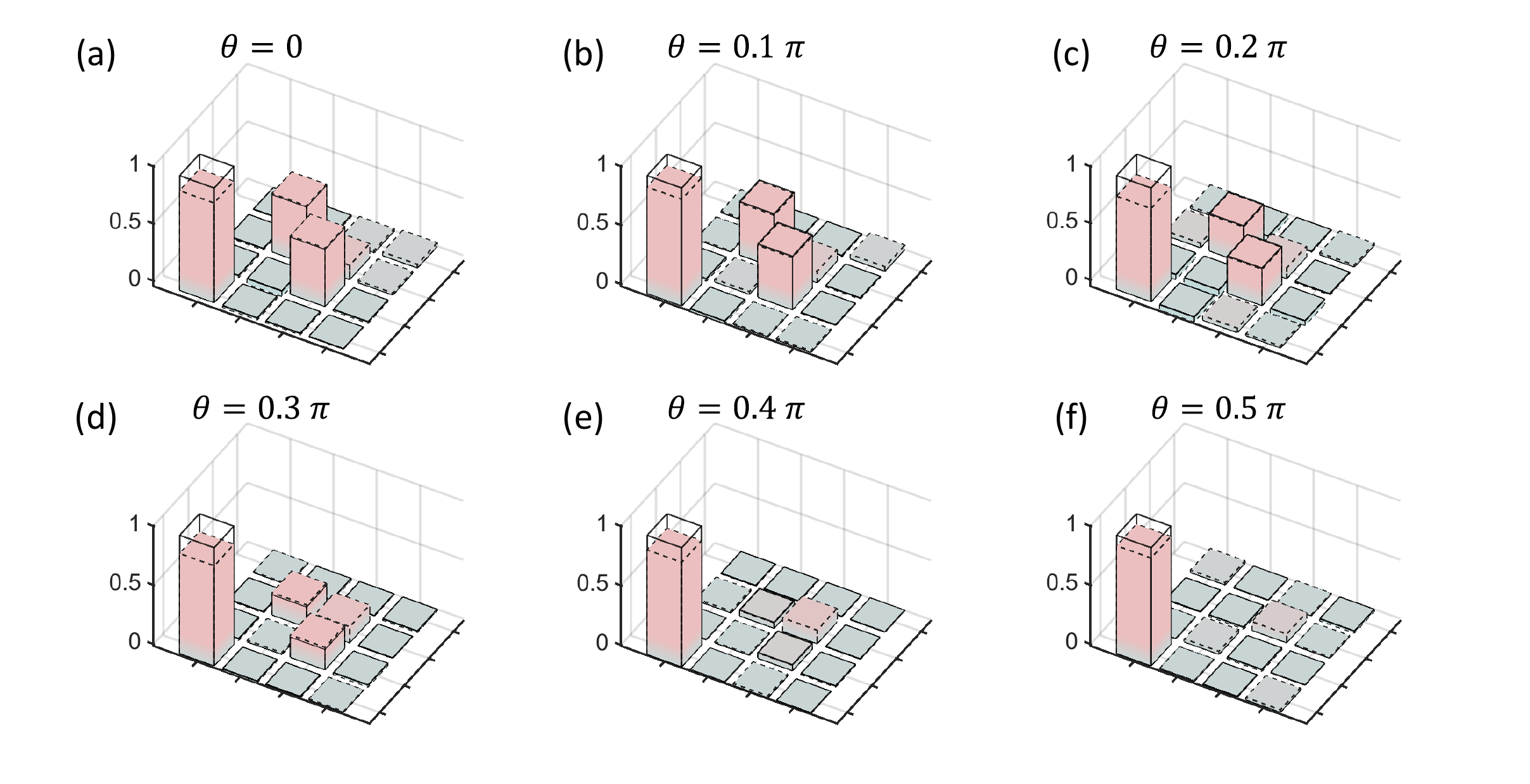}
\caption{The tomography results of fixing $p=1$ in system $\rho$.}
\label{exptomores1}
\end{figure}

\begin{figure}
\centering
\includegraphics[width=0.9\linewidth]{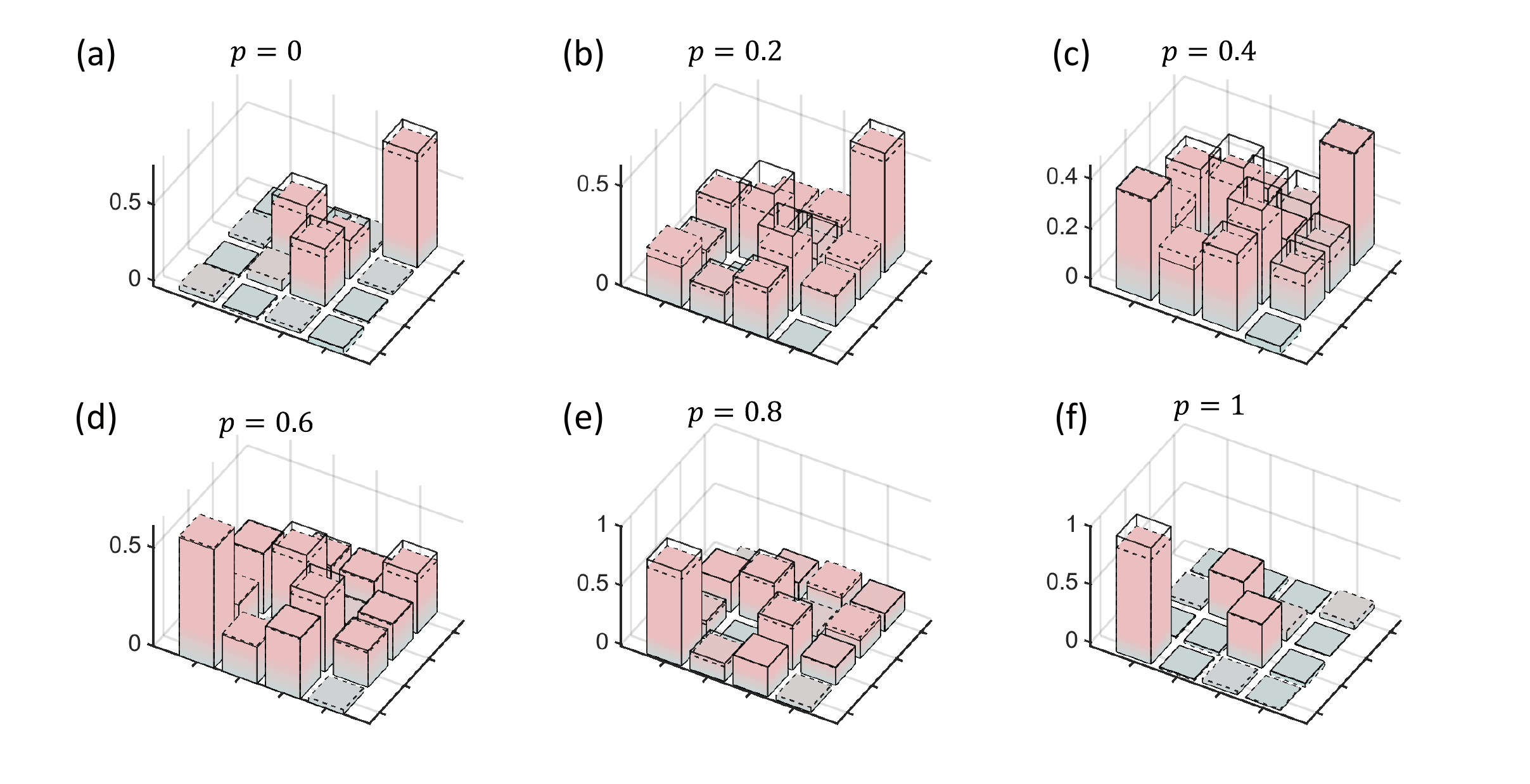}
\caption{The tomography results of fixing $\theta=\pi/6$ of channel.}
\label{exptomores2}
\end{figure}

\newpage
\section{Experimental Result of Each Unitary $U$}
In our experiment, we selected 200 Clifford Unitary operators $U$ to conduct the randomized measurements. Each point of Fig.~\ref{fig:res} is derived from 200 experimental results. Here we present the comparison between these experimental results and simulation results behind each data point, with the comparative data displayed in Fig.\ref{expres1} and Fig.\ref{expres2}.

\begin{figure}
\centering
\includegraphics[width=0.9\linewidth]{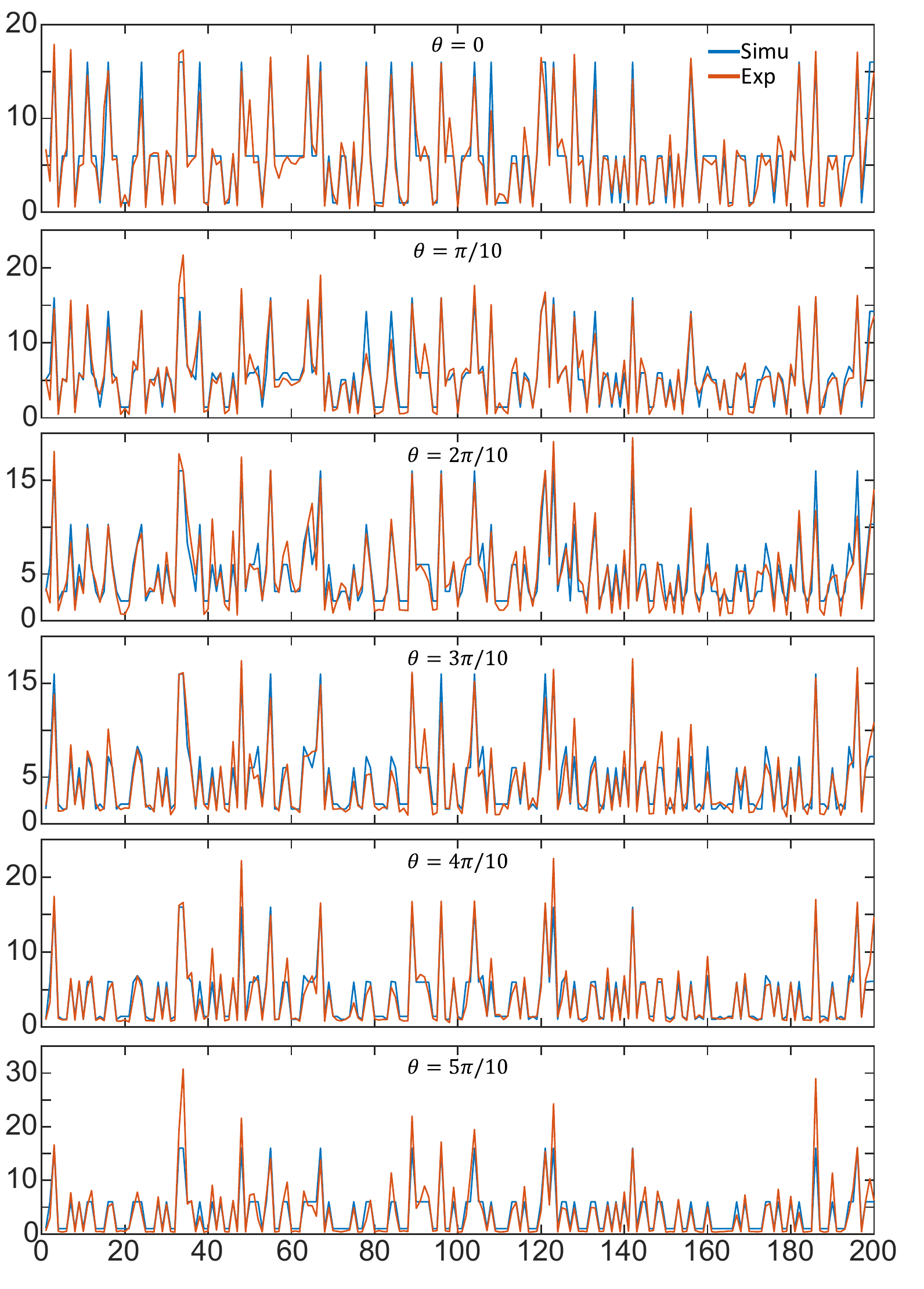}
\caption{The comparison results of fixing $p=1$ in system $\rho$.}
\label{expres1}
\end{figure}

\begin{figure}
\centering
\includegraphics[width=0.9\linewidth]{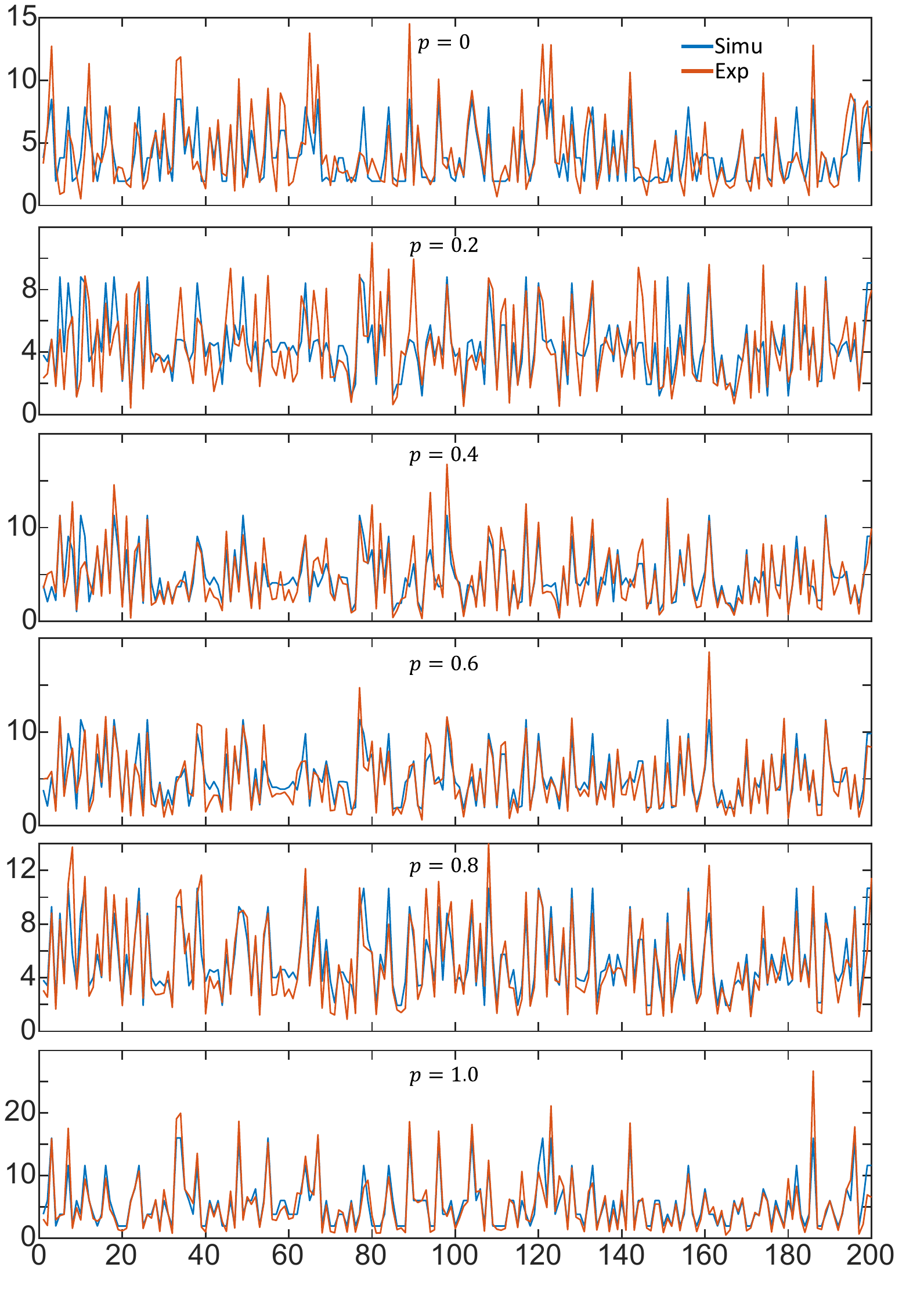}
\caption{The comparison results of fixing $\theta=\pi/6$ of channel.}
\label{expres2}
\end{figure}

\end{document}